\documentclass[11pt]{amsart}

\usepackage{amsmath,amssymb,amsthm,bm,enumerate,epstopdf,fullpage,graphicx,longtable,mathrsfs,mathtools,tabularx}
\usepackage[all]{xy}
\usepackage{geometry}                
\usepackage{chngcntr}
\usepackage{apptools}
\AtAppendix{\counterwithin{lem}{section}}

\newtheorem{lem}{Lemma}
\newtheorem{prop}{Proposition}
\newtheorem{cor}{Corollary}
\theoremstyle{remark} 
	\newtheorem{rem}{Remark}
	
\theoremstyle{definition}

\newtheorem{innercustomthm}{Propostion}
\newenvironment{customthm}[1]
  {\renewcommand\theinnercustomthm{#1}\innercustomthm}
  {\endinnercustomthm}

\newcommand{\abs}[1]{\left| #1 \right|}
\newcommand{\BigO}[1]{\mathcal{O}{\textstyle\left( #1\right)}}

\newcommand{\ie}{\textit{i.e.,}\,}
\newcommand{\eg}{\textit{e.g.,}\,}
\newcommand{\nb}{\textit{n.b.,}\, }
\newcommand{\defn}{:=}

\begin{document}

\title{Invasion Probabilities, Hitting Times, and Some Fluctuation Theory for the Stochastic Logistic Process}
\author{Todd L. Parsons}
\date{\today}

\begin{abstract}
We consider excursions for a class of stochastic processes describing a population of discrete individuals experiencing density-limited growth, such that the population has a finite carrying capacity and behaves qualitatively like the classical logistic model \cite{Verhulst1838} when the carrying capacity is large.  Being discrete and stochastic, however, our population nonetheless goes  extinct in finite time.  We present results concerning the maximum of the population prior to extinction in the large population limit, from which we obtain establishment probabilities and upper bounds for the process, as well as estimates for the waiting time to establishment and extinction.  As a consequence, we show that conditional upon establishment, the stochastic logistic process will with high probability greatly exceed carrying capacity an arbitrary number of times prior to extinction.
\end{abstract}

\maketitle

\section{Introduction}

In this note, we investigate large fluctuations for the stochastic logistic process.  For present purposes, a stochastic logistic process is a continuous-time Markov process taking values in the set of non-negative integers and such that in the large population limit, the process evolves in time according to Verhulst's logistic equation \cite{Verhulst1838} or Kolmogorov's generalization thereof \cite{Kolmogorov36}.    Our prototypical examples of this will be the birth and death processes $X^{(n)}$ with density dependent fecundity or density dependent mortality, with Markov transition rates
\begin{equation}\label{DDF}
	q^{(n)}_{i,i+1} = {\textstyle \lambda \left(1-\frac{i}{n}\right) i}
	\quad \text{and} \quad
	q^{(n)}_{i,i-1} =  \mu i
\end{equation}
and 
\begin{equation}\label{DDM}
	q^{(n)}_{i,i+1} = \lambda i
	\quad \text{and} \quad
	q^{(n)}_{i,i-1} =  {\textstyle \mu \left(1+ \frac{i}{n}\right)i}
\end{equation}	
respectively.   We intentionally take a mostly self-contained approach, that, outside of some results quoted in the introduction, does not assume specialist knowledge of stochastic processes. 
Rather than being a limitation, this will allow us to derive sharper results than are available via large deviations techniques: we will obtain not merely exponential rates, but also the constant ``prefactors''. 

Both of our prototypes \eqref{DDF}, \eqref{DDM} are examples of the density-dependent population processes introduced in \cite{Kurtz1970}.  More generally, we will consider a continuous time Markov chain, $X^{(n)}$, with transitions 
\begin{equation}\label{RATES}
	q^{(n)}_{i,i+1} = \lambda^{(n)}_{i} = {\textstyle \lambda(\frac{i}{n}) i} 
	\quad \text{and} \quad
	q^{(n)}_{i,i-1} = \mu^{(n)}_{i} = {\textstyle \mu(\frac{i}{n}) i}
\end{equation}
for non-negative, continuous functions $\lambda(x)$ and $\mu(x)$, which we assume to be bounded on compact sets.  For ease of notation, we will suppress the exponents and write $\lambda^{(n)}_{i} = \lambda_{i}$ and $\mu^{(n)}_{i} = \mu_{i}$ in what follows.

In \cite{Kurtz1970}, it is shown that if $x(t,x_{0})$ is the solution to \begin{equation}\label{LLN}
	\dot{x} = (\lambda(x)-\mu(x)) x
\end{equation}
with initial condition $x(0,x_{0}) = x_{0}$ and $\frac{1}{n} X^{(n)}(0) \to x_{0}$, then the rescaled process
$x^{(n)}(t) = \frac{1}{n} X^{(n)}(t)$ converges to $x(t,x_{0})$.  To be precise, for any fixed $T > 0$, 
\[
	\lim_{n \to \infty} \sup_{t \leq T} | x^{(n)}(t) - x(t,x_{0})| = 0 \quad \text{a.s.}
\]
This implies that our prototypes above approach -- over finite time horizons -- deterministic processes satisfying the logistic equation, 
\[
	\dot{x} =  r \left(1-\frac{x}{\kappa}\right)x
\]
where $r = \lambda-\mu$ and $\kappa = 1-\frac{\mu}{\lambda}$ for the process with density dependent fecundity, whereas $\kappa = \frac{\lambda}{\mu}-1$ for the process with density dependent mortality.  Moreover, provided $\lambda > \mu$, the carrying capacity $\kappa$ is a stable fixed point for the dynamics. 

We might thus reasonably assume that the paths of stochastic processes cleave close to the corresponding unscaled carrying capacity, $\kappa n$.  This intuition was formalized in \cite{Kurtz1971}, where it was shown that, appropriately rescaled, the fluctuations of the stochastic process $x^{(n)}(t)$ about its deterministic limit $x(t,x_{0})$ approach an Ornstein-Uhlenbeck process as $n \to \infty$: 
\begin{equation}\label{CLT}
	V^{(n)}(t) = \sqrt{n}\left(x^{(n)}(t) - x(t,x_{0})\right) \Rightarrow V(t)
\end{equation}	
where $\Rightarrow$ denotes weak convergence in the Skorokhod space $\mathbb{D}_{\mathbb{R}}[0,\infty)$\footnote{We shall not need such technical notions in what follows, but the interested reader is referred to \cite{Billingsley1968} for a very readable account of weak convergence} and $V(t)$ satisfies the It\^o stochastic differential equation
\[
	dV(t) = b'(x(t,x_{0})) V(t)\, dt + \sigma(x(t,x_{0}))\, dB(t)
\]
with 
\[
	b(x) = (\lambda(x)-\mu(x))x \quad \text{and} \quad \sigma(x) = (\lambda(x)+\mu(x))x,
\]
and $B(t)$ a standard Brownian motion.  As $t \to \infty$, $V(t)$ relaxes towards a stationary distribution that is Gaussian with mean 0 and variance $\frac{\sigma(\kappa)}{2 b'(\kappa)}$ (see 
\cite{Barbour1976}).  Informally, the fluctuations are generically of asymptotically smaller order than the carrying capacity, $\BigO{\sqrt{n}}$ \textit{vs.} $\BigO{n}$.  Nonetheless, since 0 is the only absorbing state of the Markov chain, over a sufficiently long time horizon the stochastic logistic process will necessarily have a large fluctuation: it must eventually hit zero.  

In what follows, we shall demonstrate the perhaps counterintuitive result that, even starting from a single individual, in the limit as $n \to \infty$, the stochastic logistic process (as defined by \eqref{RATES}, under suitable assumptions on the functions $\lambda$ and $\mu$ that assure ``logistic-like'' qualitative dynamics) has a non-zero probability of greatly exceeding carrying capacity (potentially fluctuations to more than double the carrying capacity).   To be precise, there is a ``potential barrier'', $\eta > \kappa$, defined in the next section, such that as $n \to \infty$, the probability $x^{(n)}(t)$ reaches any $\iota < \eta$ tends to a non-zero limit that is independent of $\iota$.  Moreover, having attained such a high value, it will as $n \to \infty$, return there an arbitrary number of times.  

Furthermore, once the stochastic process is successfully established (which for present purposes, means reaching a population size of $m_{n}$ individuals for \textit{any} fixed sequence  $\{m_{n}\}$ such that $m_{n} \to \infty$) then with high probability (\ie with probability approaching one as $n \to \infty$)  it will have at least one fluctuation far above carrying capacity (and thus arbitrarily many) prior to extinction.  On the other hand, we shall also see that there is a sharp bound on such fluctuations: for any $\iota > \eta$, the probability that $x^{(n)}(t)$ reaches $\iota$ is zero (our results are inconclusive for $\iota = \eta$).

Finally, we will obtain sharp asymptotics for the expected first hitting time of $0$, $\kappa$, and any $\iota < \eta$, as well as the return times to carrying capacity and to $\iota$, giving a relatively complete qualitative portrait of the fluctuations of the stochastic logistic process.

\section{Invasion Probabilities}

We start with our assumptions on $\lambda$ and $\mu$.  We want the deterministic process \eqref{LLN} to be competitive in the sense of \cite{Kolmogorov36}, so the individual birth rate $\lambda(x)$ and death rate $\mu(x)$ will be required to be decreasing and increasing functions of the population density $x$, respectively \footnote{In \cite{Kolmogorov36}, Kolmogorov actually makes the weaker assumption that the net per-capita growth rate $\lambda(x) - \mu(x)$ is decreasing; we make this stronger assumption to ensure that $f(x) = \ln{\frac{\mu(x)}{\lambda(x)}}$ is increasing, which is crucial to our results.  One can easily construct examples where the former holds, but not the latter, \eg $\lambda(x) = \lambda + e^{x-\kappa}$, $\mu(x) = \mu(1+x) + e^{x-\kappa}$, for $\lambda > \mu > 0$ and $\kappa = \frac{\lambda}{\mu} - 1$.}.  Further, we want \eqref{LLN} to have bounded trajectories, a unique stable fixed point, and that $x = 0$ be a repeller for the dynamics, and thus assume there is a value $\kappa > 0$ such that $\lambda(x) - \mu(x) > 0$ for $0 < x < \kappa$, $\lambda(\kappa) = \mu(\kappa)$, and $\lambda(x) - \mu(x) < 0$ for $x > \kappa$.  

We allow the possibility that $\lambda(\omega) = 0$ for some $\omega > 0$; by the above, we must have $\omega > \kappa$ \eg for the birth and death process with density dependent fecundity, $\omega = 1$, whereas for the process with density dependent mortality, $\omega = \infty$.

For any non-negative integer $m$, let
\[
	T^{(n)}_{m} = \inf\left\{t \geq 0 : X^{(n)}(t) = m \right\}
\]
and
\[
	h^{(n)}_{a,b}(m) = \mathbb{P}_{m}\left\{T^{(n)}_{a} < T^{(n)}_{b}\right\},
\]
where $\mathbb{P}_{m}$ indicates the probability conditional on $X^{(n)}(0) = m$ (similarly, we will write $\mathbb{E}_{m}$ for the expectation conditional on $X^{(n)}(0) = m$).  \textit{A priori}, if $\lambda_{m} = 0$, then $T^{(n)}_{m+1} = \infty$ and $h^{(n)}_{a,b}(m) = 1$ for $b > m > a$.

By looking at the process at its jump times (\textit{i.e.}, the embedded Markov chain), it is clear that for $a < m < b$, the probabilities $h^{(n)}_{a,b}(m)$ satisfy a recurrence relation
\[
	h^{(n)}_{a,b}(m) = \frac{\lambda_{m}}{\lambda_{m} + \mu_{m}} h^{(n)}_{a,b}(m+1) 
	+ \frac{\mu_{m}}{\lambda_{m} + \mu_{m}} h^{(n)}_{a,b}(m-1)
\]
with boundary conditions $h^{(n)}_{a,b}(a) = 1$ and $h^{(n)}_{a,b}(b) = 0$ ($\frac{\lambda_{m}}{\lambda_{m} + \mu_{m}}$ and $\frac{\mu_{m}}{\lambda_{m} + \mu_{m}}$ are the probability that, given there are $m$ individuals prior to a given jump, that jump is a birth or death, respectively).

This may be solved to yield
\begin{equation}\label{H}
	h^{(n)}_{a,b}(m) = \frac{\sum_{i=m}^{b-1} \prod_{j=1}^{i} \frac{\mu_{j}}{\lambda_{j}}}
		{\sum_{i=a}^{b-1} \prod_{j=1}^{i} \frac{\mu_{j}}{\lambda_{j}}},
\end{equation}
where we set $\sum_{i=b}^{b-1} \prod_{j=1}^{i} \frac{\mu_{j}}{\lambda_{j}} = 0$ and $\prod_{j=1}^{0} \frac{\mu_{j}}{\lambda_{j}} = 1$\footnote{
Alternately, we may observe that 
\[
	\phi(m) =  \sum_{i=1}^{m-1} \prod_{j=1}^{i} \frac{\mu_{j}}{\lambda_{j}}
\]
is a martingale for the process $X^{(n)}$, which has generator
\[
	(A^{(n)} f)(i) = \lim_{t \downarrow 0} \frac{\mathbb{E}_{i}[f(X^{(n)}(t))]-f(i)}{t} 
	= \lambda_{i} (f(i+1) - f(i)) + \mu_{i} (f(i-1)-f(i))
\]
(an easy calculation shows $A^{(n)}\phi \equiv 0$).   Given any two positive integers $a < b$, let $T^{(n)} = T^{(n)}_{a} \wedge T^{(n)}_{b}$. The optional stopping theorem tells us that 
\[
	\phi(m) = \mathbb{E}_{m} [\phi(X^{(n)}(T^{(n)}))]
	= \phi(a)  \mathbb{P}_{m}\left\{T^{(n)}_{a} < T^{(n)}_{b}\right\}
	+ \phi(b)  \mathbb{P}_{m}\left\{T^{(n)}_{b} < T^{(n)}_{a}\right\},
\]
and, since $\mathbb{P}_{m}\left\{T^{(n)}_{b} < T^{(n)}_{a}\right\} = 1 -  \mathbb{P}_{m}\left\{T^{(n)}_{a} < T^{(n)}_{b}\right\}$, we obtain
\[
	 \mathbb{P}_{m}\left\{T^{(n)}_{a} < T^{(n)}_{b}\right\} = \frac{\phi(m)-\phi(b)}{\phi(a)-\phi(b)},
\]
giving an alternate derivation of $h^{(n)}_{a,b}(m)$.}.  A standard reference for such results is \cite{Karlin1975}.   

We wish to apply this to the case when $a = 0$ and $b = \lfloor xn \rfloor$ for some $0 < x \leq \omega$, in the limit as $n \to \infty$.  To that end, we observe that 
\[
	\prod_{j=1}^{i} \frac{\mu_{j}}{\lambda_{j}} 
	= \prod_{j=1}^{i} \frac{\mu(\frac{j}{n})}{\lambda(\frac{j}{n})}
	= e^{\sum_{j=1}^{i} \ln\frac{\mu(\frac{j}{n})}{\lambda(\frac{j}{n})}}.
\]
Let $f(x) = \ln\frac{\mu(x)}{\lambda(x)}$, so $f$ is an increasing function and $f(x) \leq 0$ for $x \leq \kappa$.  Note since $f$ is increasing and $f(x) \geq 0$ for all $x \geq \kappa$, either $\omega < \infty$, or there exists $\kappa < \zeta < \infty $ such that
\[
	\int_{0}^{\zeta} f(x)\, dx = 0, 
\]
whereas 
\[
	\int_{0}^{\iota} f(x)\, dx 
\]
is positive for $\iota > \zeta$ and negative if $\iota < \zeta$.  Let $\eta = \min\{\omega,\zeta\}$.

For our choice of $f$ we then have
\[
	\prod_{j=1}^{i} \frac{\mu_{j}}{\lambda_{j}}  = e^{\sum_{j=1}^{i} f(\frac{j}{n})},
\]
and (see Lemma \ref{INEQ}),   
\[
	e^{n \int_{0}^{\frac{i}{n}} f(x)\, dx} \leq \prod_{j=1}^{i} \frac{\mu_{j}}{\lambda_{j}} 
		\leq e^{n \int_{0}^{\frac{i}{n}} f(x)\, dx + f\left(\frac{i}{n}\right) - f(0)}.
\]

By the intermediate value theorem for integrals, we have
\[
	n \int_{0}^{\frac{i}{n}} f(x)\, dx = f(z_{i,n}) i
\] 
for some $z_{i,n} \in [0,\frac{i}{n}]$.  
Now, fix $0 < \varepsilon < \kappa$. Provided $i \leq  \lfloor \iota n \rfloor$ for $0 < \iota < \eta$, either $\frac{i}{n} < \varepsilon$, in which case $f(z_{i,n}) < f(\varepsilon) < 0$, or
\[ 
	f(z_{i,n}) \varepsilon  < f(z_{i,n}) \frac{i}{n} = \int_{0}^{\frac{i}{n}} f(x)\, dx < 
		\min\left\{\int_{0}^{\varepsilon} f(x)\, dx, \int_{0}^{\iota} f(x)\, dx\right\} < 0,
\]
(\nb and thus, $\rho \defn \sup_{n} f(z_{i,n}) < 0$.

Now
\[	
	0 \leq f\left(\frac{i}{n}\right) - f(0) 
	\leq f(\iota) - f(0),
\]
so if 
\[
	a_{n,i} = \begin{cases}
		\prod_{j=1}^{i} \frac{\mu_{j}}{\lambda_{j}} & \text{if $i \leq \lfloor \iota n \rfloor$, and}\\
		0 & \text{otherwise}
	\end{cases},
\]
then $a_{n,i} \leq e^{f(\iota) - f(0)} e^{\rho i}$, and $e^{\rho} < 1$, so 
\[
	\sum_{i = 1}^{\infty} e^{f(\iota) - f(0)} e^{\rho i} 
	= \frac{e^{f(\iota) - f(0)+ \rho}}{1-e^{\rho}}.
\]
We may thus apply the dominated convergence theorem (see Lemma \ref{DCTS} in the appendix) to conclude that  
\begin{align*}
	\lim_{n \to \infty} \sum_{i=m}^{\lfloor \iota n \rfloor-1} \prod_{j=1}^{i} \frac{\mu_{j}}{\lambda_{j}} 
		&= \lim_{n \to \infty} \sum_{i=m}^{\infty} a_{n,i}\\
		&=  \sum_{i=m}^{\infty}  \lim_{n \to \infty} a_{n,i}\\
		&=  \sum_{i=m}^{\infty}  \left(\frac{\mu(0)}{\lambda(0)}\right)^{i} 
		= \frac{\left(\frac{\mu(0)}{\lambda(0)}\right)^{m}}{1-\frac{\mu(0)}{\lambda(0)}}.
\end{align*}

Applying this limit with the expression for the extinction probability, \eqref{H}, we have

\begin{prop}\label{EXTINCTION}
Let $0 < \iota < \eta$.  Then,
\[
	\lim_{n \to \infty} h^{(n)}_{0,\lfloor \iota n \rfloor}(m) = \left(\frac{\mu(0)}{\lambda(0)}\right)^{m}.
\]
\end{prop}

\begin{rem}
We note that in the limit as $n \to \infty$, the hitting probability is independent of the value of $\iota > 0$, and is equal to the probability of extinction for the pure birth and death process with rates $\lambda_{i} = \lambda(0) i$ and $\mu_{i} = \mu(0) i$.  
\end{rem}

\begin{rem}\label{TWICE}
In our example with density dependent mortality, \eqref{DDM}, $\eta > 2\kappa$:
for $x \geq 0$, $\ln(x) \leq x - 1$, so 
\[
	f(x) \leq \frac{\mu(x)}{\lambda(x)} - 1 = \frac{\mu(x) - \lambda(x)}{\lambda(x)} = 
	\frac{\mu}{\lambda}(x-\kappa),
\]
with the inequality strict except at $x = \kappa$.  Then, if $0 < \iota < 2\kappa$
\[
	\int_{0}^{\iota} f(x)\, dx 
	< \frac{\mu}{\lambda} \int_{0}^{\iota} x-\kappa\, dx 
	=  \frac{\mu}{\lambda} \left(\frac{1}{2}\iota-\kappa\right) \iota.
\]
Thus, the right hand side is less than zero and $\eta > 2\kappa$.  Thus, depending on the model, fluctuations to twice carrying capacity are possible, though not generically: a similar argument shows that in the example with density dependent mortality, $\eta < 2\kappa (= \omega)$.
\end{rem}

\begin{rem}
Finally, note that as $m \to \infty$, $h^{(n)}_{0,\lfloor \iota n \rfloor}(m) \to 0$.  Thus, there exists a sequence $m_{n} \to \infty$ such that $\lim_{n \to \infty} h^{(n)}_{0,\lfloor \iota n \rfloor}(m_{n}) \to 0$ (this is a general property of double sequences; see Lemma \ref{SUBSEQ}) \ie if the population reaches $m_{n}$, it will, with high probability reach the much higher bound $\lfloor \iota n \rfloor$,  greatly exceeding carrying capacity.
\end{rem}

In fact, we can prove a much stronger statement:

\begin{prop}\label{REACH} 
Let $\{m_{n}\}$ be any sequence such that $m_{n} \to \infty$. 
\[
	\lim_{n \to \infty} h^{(n)}_{0,\lfloor \iota n \rfloor}(m_{n}) = 0.
\]
\ie if the population reaches $m_{n}$, it will, with high probability, hit any value $\iota n$ for $\iota < \eta$ prior to extinction.
\end{prop}

\begin{proof}
The proof is morally identical to that of Proposition \ref{EXTINCTION}: set
\[
	\tilde{a}_{n,i} = \begin{cases}
		\prod_{j=1}^{i} \frac{\mu_{j}}{\lambda_{j}} 
			& \text{if $m_{n} \leq i \leq \lfloor \iota n \rfloor$, and}\\
		0 & \text{otherwise.}
	\end{cases}
\]
Then, as before, $\sum_{i = 1}^{\infty} \tilde{a}_{n,i} < \infty$, whereas for any fixed value of $i$,
\[
	\lim_{n \to \infty} \tilde{a}_{n,i} = 0.
\]
Again, the result follows by interchanging limit and sum to conclude $h^{(n)}_{0,\lfloor \iota n \rfloor}(m_{n}) \to 0$.
\end{proof}

Finally, we note that despite the fact that there are very large fluctuations above carrying capacity, the rescaled process nonetheless remains bounded with high probability:

\begin{prop}\label{BOUND}
Let $\iota > \eta$.  Then, for any sequence $m_{n}$ such that $m_{n} < \lfloor \nu n \rfloor$ for some  $\nu < \iota$,
\[
	\lim_{n \to \infty} h^{(n)}_{0,\lfloor \iota n \rfloor}(m_{n}) = 1,
\]
so $x^{(n)} = \frac{X^{(n)}}{n}$ is asymptotically stochastically bounded by $\eta+\varepsilon$ for all $\varepsilon > 0$. 
\end{prop}

\begin{proof}
We consider two cases, $\nu < \eta$, and $\nu \geq \eta$.  For the former, we observe that, as in Proposition \ref{EXTINCTION}, the sum 
\[
	\sum_{i=0}^{m_{n}-1} \prod_{j=1}^{i} \frac{\mu_{j}}{\lambda_{j}}
\] 
is bounded.  Now, for all $a > \eta$,
\[
	\prod_{j=1}^{\lceil an \rceil} \frac{\mu_{j}}{\lambda_{j}} \geq e^{n \int_{0}^{a} f(x)\, dx} > 1.
\]
In particular, for $\iota > \eta$, we have
\[
	\sum_{i=m_{n}}^{\lfloor \iota n \rfloor-1} \prod_{j=1}^{i} \frac{\mu_{j}}{\lambda_{j}}
	\geq \sum_{i= \lceil \eta n \rceil}^{\lfloor \iota n \rfloor-1} 1,
\]
and this sum diverges as $n \to \infty$.  Thus, 
\[
	h^{(n)}_{0,\lfloor \iota n \rfloor}(m_{n}) 
	= \frac{\sum_{i=m_{n}}^{\lfloor \iota n \rfloor-1} \prod_{j=1}^{i} \frac{\mu_{j}}{\lambda_{j}}}
	{\sum_{i=0}^{m_{n}-1} \prod_{j=1}^{i} \frac{\mu_{j}}{\lambda_{j}}+\sum_{i=m_{n}}^{\lfloor \iota n \rfloor-1} \prod_{j=1}^{i} \frac{\mu_{j}}{\lambda_{j}}} \to 1
\]
as $n \to \infty$.

The second case is approached similarly, only now we observe that, proceeding as in Proposition \ref{EXTINCTION}, we have that for $i < \lfloor \nu n \rfloor$, 
\[
	e^{f(\nu) i} \leq \prod_{j=1}^{i} \frac{\mu_{j}}{\lambda_{j}} \leq e^{f(\nu) - f(0) + f(\nu) i},
\]
so that 
\[
	\sum_{i=0}^{m_{n}-1} \prod_{j=1}^{i} \frac{\mu_{j}}{\lambda_{j}}
		\leq  e^{f(\nu) - f(0)} \frac{e^{n \nu f(\nu)} - 1}{e^{\nu f(\nu)} - 1},
\]
which diverges.  However, since $\mu_{j} > \lambda_{j}$ for $j > \kappa n$, we have that for $i > n \nu$,
\[ 
	\prod_{j=1}^{i} \frac{\mu_{j}}{\lambda_{j}}  =
	 \prod_{j=1}^{\lfloor \nu n \rfloor} \frac{\mu_{j}}{\lambda_{j}}
	 \prod_{\lfloor \nu n \rfloor + 1}^{i} \frac{\mu_{j}}{\lambda_{j}} \geq e^{n \nu f(\nu)},
\]
so that 
\[ 
	\sum_{i=m_{n}}^{\lfloor \iota n \rfloor-1} \prod_{j=1}^{i} \frac{\mu_{j}}{\lambda_{j}}
	  \geq  \sum_{i= \lceil \nu n \rceil}^{\lfloor \iota n \rfloor-1} e^{n \nu f(\nu)},
\]
which diverges at the asymptotically greater rate of $\BigO{n e^{n \nu f(\nu)}}$, and again $h^{(n)}_{0,\lfloor \iota n \rfloor}(m_{n}) \to 1$ as $n \to \infty$.
\end{proof}

An intuitive understanding of these results can be obtained by observing that the stochastic logistic process is equivalent to a random walk on $\mathbb{N}_{0}$ absorbed at 0 in the potential $n V^{(n)}$, where 
\[
	V^{(n)}(m) \defn \begin{cases} \frac{1}{n} \sum_{i=1}^{m} 
		\ln\frac{\mu\left(\frac{i}{n}\right)}{\lambda\left(\frac{i}{n}\right)} & \text{if $m > 0$, and }\\
		0 & \text{if $m=0$,}
	\end{cases}
\]
(we write this as  $n V^{(n)}(m)$ to emphasize that, provided $i < \omega n$, $V^{(n)}(i)$ remains bounded as $n \to \infty$).  Recall that when in state $m$, the random walk in $n V^{(n)}$ increases by 1 with probability
\[	
	p_{i} \defn \frac{e^{-n V^{(n)}(m)}}{e^{-n V^{(n)}(m-1)} + e^{-n V^{(n)}(m)}} 
	\left(= \frac{\lambda_{m}}{\lambda_{m}+\mu_{m}}\right),
\]
and decreases by 1 with probability $q_{i} \defn 1-p_{i}$; this is exactly the embedded Markov chain we used to compute $h^{(n)}_{a,b}(m)$. 

Now, we have observed that for $\iota < \omega$, as $n \to \infty$,
\[
	V^{(n)}(\iota n) \to V(\iota) \defn \int_{0}^{\iota} f(x)\, dx,
\]
and the latter has a unique minimum at $\kappa$.  Propositions \ref{REACH} and \ref{BOUND} then tell us that the process can reach any point $\iota n$ such that $V(\iota) < V(0) = 0$, but cannot reach points $\iota > \eta$, where $V(\iota) > 0$  -- \ie despite our use of quotes in introducing it, $\eta$ is truly a potential barrier for the process.  Because the potential is scaled by $n$, it becomes arbitrarily harder for the process to reach points of higher potential, a statement we quantify as Proposition \ref{HARDER} below.

\section{Refinements, Fluctuation Theory, and Hitting Times}

Using the potential, we can refine our understanding of the large excursions for the logistic process; we begin by observing that for any two points $a < b < n\omega$, and $m < b$,
\begin{equation}\label{henV}
	h^{(n)}_{a,b}(m) = \frac{\sum_{i=m}^{b-1} e^{n V^{(n)}(i)}}{\sum_{j=a}^{b-1} e^{n V^{(n)}(j)}}
\end{equation}

In what follows, we will require the functions $\lambda(x)$ and $\mu(x)$ to be at least continuously differentiable, which will allow us to use the following discrete analogue to Laplace's method to obtain asymptotic estimates of the infinite sums in the previous section in terms of the maximum of the potential. Since the asymptotic potential $V(x)$ is convex, its maximum over any interval $[\alpha,\beta]$ ($0 \leq \alpha < \beta < \eta$) occurs at one of the endpoints $\alpha$ or $\beta$, whereas the minimum at $\kappa$ is the unique local (and thus global) minimum.

In what follows, we will use Hardy-Vinogradov notation, so $f(n) \sim g(n)$ if
\[
	\lim_{n \to \infty} \frac{f(n)}{g(n)} = 1,
\]
and $f(n) \lesssim g(n)$ if
\[
	\limsup_{n \to \infty} \frac{f(n)}{g(n)} \leq 1,
\]
whereas $f(n) \ll g(n)$ if there exists a constant $C$ such that
\[	
	|f(n)| \leq C|g(n)|
\]
for all $n$.

\begin{prop}\label{LAPLACE}
Let $a_{n}$ and $b_{n}$ be sequences of non-negative integers such that $\frac{a_{n}}{n} \to \alpha$ and $\frac{b_{n}}{n} \to \beta$, and suppose that $\psi(x)$ and $g(x)$ are, respectively, a continuously differentiable function and a continuous function on an open interval containing $[\alpha,\beta]$,
and that $\epsilon_{n}$ is a sequence of functions on the set of integers $\{a_{n},a_{n}+1,\ldots,b_{n}\}$, uniformly converging to 0.
\begin{enumerate}[(i)] 
\item If $\psi(\alpha) > \psi(x) $ for all $\alpha < x \leq \beta$ and $\psi'(\alpha) < 0$ then
\[
	\sum_{i = a_{n}}^{b_{n}-1} (1+\epsilon_{n}(i)) g\left(\frac{i}{n}\right)
		e^{n \psi\left(\frac{i}{n}\right)}
		\sim \frac{g\left(\frac{a_{n}}{n}\right) e^{n \psi\left(\frac{a_{n}}{n}\right)}}
		{1-e^{\psi'\left(\frac{a_{n}}{n}\right)}}
\]
\item If, on the other hand, $\psi(\beta) > \psi(x) $ for all $\alpha \leq x < \beta$ and $\psi'(\beta) > 0$ then
\[
	\sum_{i = a_{n}}^{b_{n}-1} (1+\epsilon_{n}(i)) g\left(\frac{i}{n}\right)
		e^{n \psi\left(\frac{i}{n}\right)} \sim
		\frac{g\left(\frac{b_{n}}{n}\right) e^{n \psi\left(\frac{b_{n}}{n}\right)}}
		{1-e^{-\psi'\left(\frac{b_{n}}{n}\right)}}.
\]
\item Finally, if $\psi$ is twice continuously differentiable, and there exists $\gamma \in (\alpha,\beta)$ with such that $\psi(\gamma) > \psi(x) $ for all $\alpha \leq x < \beta$, $\psi'(\gamma) = 0$, and $\psi''(\gamma) < 0$, then
\[
	\sum_{i = a_{n}}^{b_{n}-1} (1+\epsilon_{n}(i)) g\left(\frac{i}{n}\right)
		 e^{n \psi\left(\frac{i}{n}\right)} \sim 
	g(\gamma) e^{n \psi(\gamma)}\sqrt{\frac{2 n \pi}{|\psi''(\gamma)|}}
\]
\end{enumerate}
\end{prop}

\begin{proof}
We will first prove the first statement.  The proof of the second is identical.

Fix $\varepsilon > 0$ such that $\psi'(\alpha) + \varepsilon < 0$.  Using Taylor's theorem, we may write 
\[
	\psi(x) = \psi(y) + \psi'(y)(x-y) + R(x,y)(x-y),
\]
where $R(x,y) \to 0$ as $|x-y| \to 0$.  Fix $\delta > 0$ such that 
\[
	|\psi'(x) - \psi'(\alpha)| < \frac{\varepsilon}{2} 
\]
for all $x$ such that $|x - \alpha| < \delta$ and 
\[
	|R(x,\alpha)| < \frac{\varepsilon}{2} \quad \text{and} \quad |g(x) - g(\alpha)| < \frac{\varepsilon}{2} 
\]
for all $x$ such that $x - \alpha < 3\delta$, and choose $\eta > 0$ such that $\psi(x) < \psi(\alpha) - \eta$ for all $x$ such that $x - \alpha \geq \delta$.  Fix $M$ and $m$ such that $|g(x)| < M$ and $|h(x)| < m$ for all $x \in [\alpha-\delta,\beta+\delta]$ and all $n$.  Since $\frac{a_{n}}{n} \to \alpha$ and $\frac{b_{n}}{n} \to \beta$, without loss of generality, we may assume that  $|\frac{a_{n}}{n} - \alpha| < \delta$, $|\frac{b_{n}}{n} - \beta| < \delta$ and $|\epsilon_{n}| < \varepsilon$ for all $n$. 

Then, 
\begin{multline*}
	\sum_{i = a_{n}}^{b_{n}-1} (1+\epsilon_{n}(i)) g\left(\frac{i}{n}\right)
		e^{n \psi\left(\frac{i}{n}\right)}\\
	= e^{n \psi\left(\frac{a_{n}}{n}\right)} \sum_{i = a_{n}}^{b_{n}-1} (1+\epsilon_{n}(i)) g\left(\frac{i}{n}\right)
		e^{n \left(\psi\left(\frac{i}{n}\right)-\psi\left(\frac{a_{n}}{n}\right)\right)}\\
	= e^{n \psi\left(\frac{a_{n}}{n}\right)}
	\left(
		\sum_{i = a_{n}}^{a_{n} + \lceil 2n \delta\rceil-1} 
		(1+\epsilon_{n}(i)) g\left(\frac{i}{n}\right) e^{n \left(\psi\left(\frac{i}{n}\right)-\psi\left(\frac{a_{n}}{n}\right)\right)}\right.\\
		+ \left. \sum_{i = a_{n} + \lceil 2n \delta\rceil}^{b_{n}-1}
		(1+\epsilon_{n}(i)) g\left(\frac{i}{n}\right) e^{n \left(\psi\left(\frac{i}{n}\right)-\psi\left(\frac{a_{n}}{n}\right)\right)}
	\right),
\end{multline*}
and, provided $i \geq a_{n} + \lceil 2n \delta \rceil$, then $\frac{i}{n} \geq \alpha + \delta$, and 
\[
	\abs{\sum_{i = a_{n} + \lceil 2n \delta\rceil}^{b_{n}-1}
		(1+\epsilon_{n}(i)) g\left(\frac{i}{n}\right)
			 e^{n \left(\psi\left(\frac{i}{n}\right)-\psi\left(\frac{a_{n}}{n}\right)\right)}}
		\leq (b_{n} - a_{n}) M(1+\varepsilon) e^{-n \eta} \to 0
\]
as $n \to \infty$, whereas if $i < a_{n} + \lceil 2n \delta \rceil$, then $\frac{i}{n} < \alpha + 3\delta$, and 
\begin{multline*}
	(1-\varepsilon)\left(g\left(\frac{a_{n}}{n}\right) - \varepsilon\right)  
	\sum_{i = a_{n}}^{a_{n} + \lceil 2n \delta\rceil-1} 
	e^{n (\psi'\left(\frac{a_{n}}{n}\right)-\frac{\varepsilon}{2})\left(\frac{i}{n}-\frac{a_{n}}{n}\right)} 
	\leq \sum_{i = a_{n}}^{a_{n} + \lceil 2n \delta\rceil-1} 
	(1+\epsilon_{n}(i)) g\left(\frac{i}{n}\right) 
		e^{n \left(\psi\left(\frac{i}{n}\right)-\psi\left(\frac{a_{n}}{n}\right)\right)}\\
	\leq (1+\varepsilon)\left(g\left(\frac{a_{n}}{n}\right) + \varepsilon\right)  
	\sum_{i = a_{n}}^{a_{n} + \lceil 2n \delta\rceil-1} 
		e^{n(\psi'\left(\frac{a_{n}}{n}\right)+\frac{\varepsilon}{2})\left(\frac{i}{n}-\frac{a_{n}}{n}\right)}, 
\end{multline*}
and
\[
	\sum_{i = a_{n}}^{a_{n} + \lceil 2n \delta\rceil-1} 
		e^{n (\psi'\left(\frac{a_{n}}{n}\right)-\varepsilon)\left(\frac{i}{n}-\frac{a_{n}}{n}\right)} 
	=  \sum_{i = 0}^{\lceil 2n \delta\rceil-1} e^{(\psi'\left(\frac{a_{n}}{n}\right)-\varepsilon)i}
	= \frac{e^{(\psi'\left(\frac{a_{n}}{n}\right)-\varepsilon)\lceil 2n \delta\rceil}-1}
		{e^{(\psi'\left(\frac{a_{n}}{n}\right)-\varepsilon)}-1}.
\]
We now observe that $|\psi'\left(\frac{a_{n}}{n}\right) - \psi'(\alpha)| < \frac{\varepsilon}{2}$, so
$\psi'\left(\frac{a_{n}}{n}\right)+\frac{\varepsilon}{2} <  \psi'(\alpha) + \varepsilon < 0$ and
\[
	e^{(\psi'\left(\frac{a_{n}}{n}\right) + \frac{\varepsilon}{2})\lceil 2n \delta\rceil -1} \to 0
\]
as $n \to \infty$.
Proceeding similarly we obtain a lower bound.  

Since $\varepsilon > 0$ can be chosen arbitrarily small, the result follows.

To prove the third statement, we proceed as previously and write
\[
	\psi(x) = \psi(y) + \psi'(y)\left(x-y\right) + \left(\frac{1}{2}\psi'(y)+ R(x,y)\right)(x-y)^{2}
\]
where $R(x,y) \to 0$ as $|x-y| \to 0$.  Fix $\varepsilon > 0$ sufficiently small that $\psi''(\gamma) + \varepsilon <0$, and choose $\delta > 0$ sufficiently small that $|R(x,y)| < \varepsilon$ and $|g(x)-g(y)| < \varepsilon$ for all $|x-y| < 2\delta$.  As before, suppose that $|g(x)| < M$ and  $|h(x)| < m$  for $x \in [\alpha-\delta,\beta+\delta]$, that $\psi(\gamma) > \psi(x) + \eta$ for $|\gamma - x| > \delta$ and that $\epsilon_{n} < \frac{\varepsilon}{2m}$ for all $n$.  Then,
\begin{multline*}
	\sum_{i = a_{n}}^{b_{n}-1} (1+\epsilon_{n}(i)) g\left(\frac{i}{n}\right)
		e^{n \psi\left(\frac{i}{n}\right)}\\
	= e^{n \psi\left(\frac{\lfloor \gamma n \rfloor}{n}\right)} \sum_{i = a_{n}}^{b_{n}-1} (1+\epsilon_{n}(i)) g\left(\frac{i}{n}\right)
		e^{n \left(\psi\left(\frac{i}{n}\right)-\psi\left(\frac{\lfloor \gamma n \rfloor}{n}\right)\right)}\\
	= e^{n \psi\left(\frac{\lfloor \gamma n \rfloor}{n}\right)}
	\left(
		\sum_{i = a_{n}}^{\lfloor \gamma n \rfloor - \lceil n \delta\rceil-1} 
		(1+\epsilon_{n}(i)) g\left(\frac{i}{n}\right) e^{n \left(\psi\left(\frac{i}{n}\right)-\psi\left(\frac{\lfloor \gamma n \rfloor}{n}\right)\right)}\right.\\
		+\left. \sum_{i = \lfloor \gamma n \rfloor - \lceil n \delta\rceil}
			^{\lfloor \gamma n \rfloor + \lceil n \delta\rceil}
		(1+\epsilon_{n}(i)) g\left(\frac{i}{n}\right) e^{n \left(\psi\left(\frac{i}{n}\right)-\psi\left(\frac{\lfloor \gamma n \rfloor}{n}\right)\right)}\right.\\
		\left. + \sum_{i = \lfloor \gamma n \rfloor + \lceil n \delta\rceil+1}^{b_{n}-1}
		(1+\epsilon_{n}(i)) g\left(\frac{i}{n}\right) e^{n \left(\psi\left(\frac{i}{n}\right)-\psi\left(\frac{\lfloor \gamma n \rfloor}{n}\right)\right)}
	\right),
\end{multline*}
where, as before, the first and last sums are bounded above by $(b_{n} - a_{n}) (1+\varepsilon)M e^{-n \eta}$ and
\begin{multline*}
	(1-\varepsilon)\left(g\left(\frac{\lfloor \gamma n \rfloor}{n}\right) - \varepsilon\right)  
	\sum_{i = \lfloor \gamma n \rfloor-\lceil n \delta\rceil}
		^{\lfloor \gamma n \rfloor + \lceil n \delta\rceil} 
		e^{n\left(\psi'\left(\frac{\lfloor \gamma n \rfloor}{n}\right)
			\left(\frac{i}{n}-\frac{\lfloor \gamma n \rfloor}{n}\right) 
		+ \left(\frac{1}{2}\psi''\left(\frac{\lfloor \gamma n \rfloor}{n}\right)-\frac{\varepsilon}{2}\right)
			\left(\frac{i}{n}-\frac{\lfloor \gamma n \rfloor}{n}\right)^{2}\right)}\\
	\leq \sum_{i = \lfloor \gamma n \rfloor-\lceil n \delta\rceil}
		^{\lfloor \gamma n \rfloor + \lceil n \delta\rceil}
		(1+\epsilon_{n}(i)) g\left(\frac{i}{n}\right) 
		e^{n \left(\psi\left(\frac{i}{n}\right)-\psi\left(\frac{\lfloor \gamma n \rfloor}{n}\right)\right)}\\
	\leq 
	(1+\varepsilon)\left(g\left(\frac{\lfloor \gamma n \rfloor}{n}\right) + \varepsilon\right)  
	\sum_{i = \lfloor \gamma n \rfloor-\lceil n \delta\rceil}^{\lfloor \gamma n \rfloor + \lceil n \delta\rceil}
		e^{n\left(\psi'\left(\frac{\lfloor \gamma n \rfloor}{n}\right)
			\left(\frac{i}{n}-\frac{\lfloor \gamma n \rfloor}{n}\right) 
		+ \left(\frac{1}{2}\psi''\left(\frac{\lfloor \gamma n \rfloor}{n}\right)+\frac{\varepsilon}{2}\right)
			\left(\frac{i}{n}-\frac{\lfloor \gamma n \rfloor}{n}\right)^{2}\right)}, 
\end{multline*}

As previously, we will show that for $n$ sufficiently large, the upper sum has an upper bound arbitrarily close to $g(\gamma) e^{n \psi(\gamma)}\sqrt{\frac{2n\pi}{|\psi''(\gamma)|}}$, and remark that the lower sum is treated identically.  Proceeding, we have
\begin{multline*}
	\sum_{i = \lfloor \gamma n \rfloor-\lceil 2n \delta\rceil}^{\lfloor \gamma n \rfloor + \lceil 2n \delta\rceil}
		e^{n\left(\psi'\left(\frac{\lfloor \gamma n \rfloor}{n}\right)
			\left(\frac{i}{n}-\frac{\lfloor \gamma n \rfloor}{n}\right) 
		+ \left(\frac{1}{2}\psi''\left(\frac{\lfloor \gamma n \rfloor}{n}\right)+\frac{\varepsilon}{2}\right)
			\left(\frac{i}{n}-\frac{\lfloor \gamma n \rfloor}{n}\right)^{2}\right)}\\
	= \sum_{i = -\lceil 2n \delta\rceil}^{\lceil 2n \delta\rceil}
		e^{\psi'\left(\frac{\lfloor \gamma n \rfloor}{n}\right) i
		+ \frac{\psi''\left(\frac{\lfloor \gamma n \rfloor}{n}\right)+\frac{\varepsilon}{2}}{2n} i^{2}},
\end{multline*}
and since $\psi'$ is continuously differentiable, it is also Lipschitz continuous on $[\alpha,\beta]$, and there is thus a constant $L$ such that 
\[
	{\textstyle \left|\psi'\left(\frac{\lfloor \gamma n \rfloor}{n}\right)\right| 
		=  \left|\psi'\left(\frac{\lfloor \gamma n \rfloor}{n}\right)- \psi'(\gamma) \right|
		\leq L \left|\frac{\lfloor \gamma n \rfloor}{n} - \gamma \right| \leq \frac{L}{n}}
\]
and 
\begin{multline*}
	e^{- 2L\delta}  \sum_{i = -\lceil 2n \delta\rceil}^{\lceil 2n \delta\rceil}
		e^{\frac{\psi''\left(\frac{\lfloor \gamma n \rfloor}{n}\right)+\frac{\varepsilon}{2}}{n} i^{2}}
	\leq  \sum_{i = -\lceil 2n \delta\rceil}^{\lceil 2n \delta\rceil}
		e^{\psi'\left(\frac{\lfloor \gamma n \rfloor}{n}\right) i
		+ \frac{\psi''\left(\frac{\lfloor \gamma n \rfloor}{n}\right)+\frac{\varepsilon}{2}}{n} i^{2}}\\
	\leq e^{2L\delta + \frac{L}{n}}  \sum_{i = -\lceil 2n \delta\rceil}^{\lceil 2n \delta\rceil}
		e^{\frac{\psi''\left(\frac{\lfloor \gamma n \rfloor}{n}\right)+\varepsilon}{n} i^{2}}
\end{multline*}

To complete the result, we require the following lemma

\begin{lem}
Let $z > 0$.  Then,
\[
	\lim_{z \to 0} \frac{\sum_{i=-\infty}^{\infty} e^{-z i^{2}}}{\sqrt{\frac{\pi}{z}}} = 1.
\]
\end{lem}

\begin{proof}
We will prove this via Poisson's summation formula \cite{Katznelson1976}, which tells us that for an integrable function $f$ with Fourier transform $\hat{f}$,
\[
	\sum_{i=-\infty}^{\infty} f(i) = \sum_{i=-\infty}^{\infty} \hat{f}(i).
\]
Applying this with $f(x) = e^{-z x^{2}}$ gives
\[
	\sum_{i=-\infty}^{\infty} e^{-z i^{2}} 
		= \sum_{i=-\infty}^{\infty} \sqrt{\frac{\pi}{z}} e^{-\frac{i^{2}}{4z}},
\]
and the result follows on observing that for $i \neq 0$, 
\[
	\lim_{z \to 0} e^{-\frac{i^{2}}{4z}} = 0.
\]
\end{proof}

Let $z_{n} = - \frac{\psi''(\gamma)+\varepsilon}{2 n}$, so $z_{n} > 0$ and $z_{n} \to 0$ as $n \to \infty$. Since $\varepsilon$ and $\delta$ may be chosen arbitrarily small, the result follows provided 
\[
	\sum_{i = - \lfloor \delta n \rfloor}^{\lfloor \delta n \rfloor} e^{-z_{n} i^{2}}  \sim 
		\sum_{i=-\infty}^{\infty}  e^{-z_{n} i^{2}}.
\]
To see the latter, we first observe that for $q < 1$,
\[
	0 < \sum_{i=m}^{\infty} e^{-z_{n} i^{2}} 
	= e^{-z_{n} m^{2}} \sum_{i=0}^{\infty} e^{-z_{n} \left((m+i)^{2}-m^{2}\right)}
	=  e^{-z_{n} m^{2}} \sum_{i=0}^{\infty} e^{-z_{n} (i^{2}+2 i m)} 
	< e^{-z_{n} m^{2}} \sum_{i=0}^{\infty} e^{-z_{n} i^{2}} 
\]
and, similarly,
\[
	0 < \sum_{i=-\infty}^{-m} e^{-z_{n} i^{2}}  < e^{-z_{n} m^{2}} \sum_{i=0}^{\infty} e^{-z_{n} i^{2}} 
\]
Thus, 
\[
	0 <  \sum_{i=-\infty}^{-m} e^{-z_{n} i^{2}} + \sum_{i=m}^{\infty} e^{-z_{n} i^{2}} 
	< e^{-z_{n} m^{2}} \left(1+ \sum_{i=-\infty}^{\infty}  e^{-z_{n} i^{2}}\right)
\]
so that 
\[
	0 < 1 -  \frac{\sum_{i = - \lfloor \delta n \rfloor}^{\lfloor \delta n \rfloor} e^{-z_{n} i^{2}}}
		{\sum_{i=-\infty}^{\infty}  e^{-z_{n} i^{2}}} 
		< e^{-z_{n} \lfloor \delta n \rfloor^{2}} \left(1 +  \frac{1}{\sum_{i=-\infty}^{\infty}  e^{-z_{n} i^{2}}}\right), 
\]
and the latter decays super-exponentially fast in $n$.  
\end{proof}

\begin{rem}
While it is appealing to observe that 
\[
	\frac{1}{n} \sum_{i = a_{n}}^{b_{n}-1} 
		(1+\epsilon_{n}(i)) g\left(\frac{i}{n}\right) e^{n \psi\left(\frac{i}{n}\right)}
\]
is essentially the Riemann sum for
\[
	\int_{\alpha}^{\beta} g(x) e^{n \psi(x)}\, dx,
\]
and then invoke the continuous form of Laplace's method, $(i)$ and $(ii)$ show that whilst the discrete and continuous results are identical for an interior maximum, they do not agree when $\psi$ has its maximum at one of the endpoints, thus invalidating this ``proof'' of Proposition \ref{LAPLACE}.
\end{rem}


\begin{cor}\label{LAPLACE2}
Let $a_{n}$ and $b_{n}$ be as above.  Then,
\begin{enumerate}[(i)] \label{I}
\item If $V(\alpha) < V(\beta)$, then
\[
	\sum_{i = a_{n}}^{b_{n}-1} e^{V^{(n)}(i)}
		\sim \sqrt{\frac{\mu\left(\frac{a_{n}}{n}\right)}{\lambda\left(\frac{a_{n}}{n}\right)}
		\frac{\lambda(0)}{\mu(0)}} \frac{e^{n V\left(\frac{a_{n}}{n}\right)}}
		{1-\frac{\mu\left(\frac{a_{n}}{n}\right)}{\lambda\left(\frac{a_{n}}{n}\right)}}
\]
\item If $V(\beta) < V(\alpha)$, then
\[
	\sum_{i = a_{n}}^{b_{n}-1} e^{V^{(n)}(i)}
		\sim \sqrt{\frac{\mu\left(\frac{b_{n}}{n}\right)}{\lambda\left(\frac{b_{n}}{n}\right)}
		\frac{\lambda(0)}{\mu(0)}} \frac{e^{n V\left(\frac{b_{n}}{n}\right)}}
		{1-\frac{\lambda\left(\frac{b_{n}}{n}\right)}{\mu\left(\frac{b_{n}}{n}\right)}}.
\]
\item In particular, if $a_{n} = o(n)$, then $\alpha = 0$ and, if $\beta < \eta$, then
\[
	\sum_{i = a_{n}}^{b_{n}-1} e^{V^{(n)}(i)}
	\sim\frac{\left(\frac{\mu(0)}{\lambda(0)}\right)^{a_{n}}}{1-\frac{\mu(0)}{\lambda(0)}}.
\]
\end{enumerate}
\end{cor}
 
\begin{proof}
Let $\psi(x) = V(x)$, $g(x) = \sqrt{\frac{\mu(x)}{\lambda(x)}\frac{\lambda(0)}{\mu(0)}}$ and
\[
	\epsilon_{n}(i) = e^{n V^{(n)}(i) - n V\left(\frac{i}{n}\right) 
		- \frac{1}{2}\left(f\left(\frac{i}{n}\right)-f(0)\right)}-1.
\]
Then $g(x)$ is continuous,
\[ 	
	 e^{n V^{(n)}(i)} = (1+\epsilon_{n}(i)) g\left(\frac{i}{n}\right) e^{n \psi\left(\frac{i}{n}\right)},
\]
and, from Lemma \ref{INEQ}, for any positive integers $a < b$, 
\[
	|\epsilon_{n}(i)| < \frac{\sup_{x \in \left[\frac{a}{n},\frac{b}{n}\right]} |f''(x)| (b-a)}{n^{3}}.
\]
The first two assertions then follow from the corresponding parts of the Proposition.

The third statement follows immediately upon observing that 
\[
	e^{n V'\left(\frac{a_{n}}{n}\right)} 
	\sim e^{V'(0)a_{n}} = \left(\frac{\mu(0)}{\lambda(0)}\right)^{a_{n}}.
\]
\end{proof}
 
\begin{rem}
We remark that Corollary \ref{LAPLACE2} gives an alternate proof of Proposition \ref{EXTINCTION}: using the first part, we see that the sum in the numerator of $h^{(n)}_{0,\lfloor \iota n \rfloor}(m)$ is asymptotic to 
$\frac{\left(\frac{\mu(0)}{\lambda(0)}\right)^{m}}{1-\frac{\mu(0)}{\lambda(0)}},$
whereas the sum in the denominator is asymptotic to  $\frac{1}{1-\frac{\mu(0)}{\lambda(0)}}$.
\end{rem}

Indeed, as an immediate consequence, we have the following refinement of Propositions \ref{REACH} and \ref{BOUND}:

\begin{customthm}{\ref{REACH}'}
Let $\kappa < \iota < \eta$, and let  $\iota' < \kappa$ be such that $V(\iota') = V(\iota)$.  Suppose $m_{n} \to \infty$ as $n \to \infty$.  Then, if $m_{n} \ll n$,
\[
	h^{(n)}_{0,\lfloor \iota n \rfloor}(m_{n}) \sim e^{V'(0) m_{n}} 
		= \left(\frac{\mu(0)}{\lambda(0)}\right)^{m_{n}},
\]
whereas if $\frac{m_{n}}{n} \to \nu \in(0,\iota)$,
\[
	h^{(n)}_{0,\lfloor \iota n \rfloor}(m_{n}) \sim \begin{cases} 
			\frac{1-\frac{\mu(0)}{\lambda(0)}}{1-\frac{\mu(\nu)}{\lambda(\nu)}} 
			\sqrt{\frac{\mu(\nu)}{\lambda(\nu)}\frac{\lambda(0)}{\mu(0)}} e^{nV(\nu)} 
			& \text{if $0 < \nu < \iota'$, and}\\
		 \frac{1-\frac{\mu(0)}{\lambda(0)}}{1-\frac{\lambda(\iota)}{\mu(\iota)}} 
			\sqrt{\frac{\mu(\iota)}{\lambda(\iota)}\frac{\lambda(0)}{\mu(0)}} e^{nV(\iota)} 
			& \text{if $\iota' < \nu < \iota$.}
		\end{cases}
\]
Obviously, if $\frac{m_{n}}{n} \to \nu > \iota$, $h^{(n)}_{0,\lfloor \iota n \rfloor}(m_{n}) = 0$.
\end{customthm}

\begin{rem}
More generally, in the sequel, given  $x < \eta$, we will define $x' <  \eta$ to be the unique value $x' \neq x$ such that $V(x') = V(x)$ (thus, $x' < \kappa$ if $x > \kappa$, and $x' > \kappa$  if $x < \kappa$).
\end{rem}

\begin{customthm}{\ref{BOUND}'}
Let $\iota > \eta$.  Suppose $m_{n} \to \infty$ as $n \to \infty$.  Then, if $m_{n} \ll n$,
\[
	h^{(n)}_{0,\lfloor \iota n \rfloor}(m_{n}) 
		\sim 1 - \sqrt{\frac{\lambda(\iota)}{\mu(\iota)}\frac{\mu(0)}{\lambda(0)}}
			\frac{1-\frac{\lambda(\iota)}{\mu(\iota)}}{1-\frac{\mu(0)}{\lambda(0)}} 	
			\left(1-\left(\frac{\mu(0)}{\lambda(0)}\right)^{m_{n}}\right)
			 e^{-nV(\iota)} 
\]
whereas if $\frac{m_{n}}{n} \to \nu < \iota$,
\[
	h^{(n)}_{0,\lfloor \iota n \rfloor}(m_{n}) \sim \begin{cases} 
			 1 - \sqrt{\frac{\lambda(\iota)}{\mu(\iota)}\frac{\mu(0)}{\lambda(0)}}
			\frac{1-\frac{\lambda(\iota)}{\mu(\iota)}}{1-\frac{\mu(0)}{\lambda(0)}} 	
			 e^{-nV(\iota)}
			& \text{if $0 < \nu < \eta$, and}\\
		 1 - \sqrt{\frac{\mu(\nu)}{\lambda(\nu)}\frac{\lambda(\iota)}{\mu(\iota)}}
			\frac{1-\frac{\lambda(\iota)}{\mu(\iota)}}{1-\frac{\mu(\nu)}{\lambda(\nu)}} 	
			 e^{n(V(\nu)-V(\iota))}
			& \text{if $\eta < \nu < \iota$.}
		\end{cases}
\]
\end{customthm}


More generally, we have that

\begin{prop}\label{HARDER}
Suppose that $\xi < \iota < \upsilon$.  Then,
\[
	\mathbb{P}_{\lfloor \iota n \rfloor}\left\{T_{\lfloor \xi n \rfloor} < T_{\lfloor \upsilon n \rfloor}\right\} 
	\sim \begin{cases}
	\sqrt{\frac{\mu(\iota)}{\lambda(\iota)}\frac{\lambda(\xi)}{\mu(\xi)}}
			\frac{1-\frac{\lambda(\xi)}{\mu(\xi)}}{1-\frac{\mu(\iota)}{\lambda(\iota)}} 	
			 e^{n(V(\iota)-V(\xi))}
	& \text{if $\xi < \iota < \min\{\upsilon,\upsilon'\}$,}\\
	\sqrt{\frac{\mu(\upsilon)}{\lambda(\upsilon)}\frac{\lambda(\xi)}{\mu(\xi)}}
			\frac{1-\frac{\lambda(\xi)}{\mu(\xi)}}{1-\frac{\mu(\upsilon)}{\lambda(\upsilon)}} 	
			 e^{n(V(\upsilon)-V(\xi))}
	& \text{if $\upsilon < \iota < \upsilon'$,}\\
	1 - \sqrt{\frac{\lambda(\upsilon)}{\mu(\upsilon)}\frac{\mu(\xi)}{\lambda(\xi)}}
			\frac{1-\frac{\mu(\upsilon)}{\lambda(\upsilon)}}{1-\frac{\lambda(\xi)}{\mu(\xi)}} 	
			 e^{n(V(\xi)-V(\upsilon))}
	& \text{if $\xi < \iota < \xi'$, and}\\
	1 - \sqrt{\frac{\lambda(\upsilon)}{\mu(\upsilon)}\frac{\mu(\iota)}{\lambda(\iota)}}
			\frac{1-\frac{\mu(\upsilon)}{\lambda(\upsilon)}}{1-\frac{\lambda(\iota)}{\mu(\iota)}}
			 e^{n(V(\iota)-V(\upsilon))}
	& \text{if $\max\{\xi,\xi'\} < \iota$.}
	\end{cases}
\] 
\end{prop}

\begin{rem} 
In the first two cases, $V(\xi) > V(\upsilon)$, whereas in the latter two, $V(\xi) < V(\upsilon)$; thus, as we informally observed previously, the probability of reaching points of higher potential is exponentially small, and thus ``harder'' than first returning to points of lower potential, which occurs with high probability.
\end{rem}

\subsection{Some Fluctuation Theory}

We start by introducing some new notation:  for any set of nonnegative integers $A$, let
\[
	T^{(n)}_{A} = \inf\left\{t \geq 0 : X^{(n)}(t) \in A \right\},
\]
be the first hitting time of $A$ and for any nonnegative $m$, let 
\[
	T^{(n)}_{m+} = \inf\left\{t \geq T^{(n)}_{\mathbb{N}_{0} - \{m\}} : X^{(n)}(t) = m\right\}
\]
be the first return time to $m$ (note that for $X^{(n)}(0) \neq m$, $T^{(n)}_{m+} = T^{(n)}_{m}$).  Let 
$N^{(n)}_{m}(t)$ be the number of visits of $X^{(n)}$ to $m$ prior to time $t$, and let $S^{(n)}_{m}(t)$ be the total time spent in state $m$ prior to time $t$.

We start with a result that is both of interest for its own sake and that will be used in subsequent derivations:

\begin{prop}
Let $\iota < \eta$ and let $z < \eta$. Then, 
\begin{multline*}
	\mathbb{P}_{\lfloor \iota n \rfloor}\left\{T^{(n)}_{\lfloor z n \rfloor} 
		< T^{(n)}_{\lfloor \iota n \rfloor+} \right\} \\
	\sim \begin{cases}
	\frac{\mu(\iota)}{\lambda(\iota)+\mu(\iota)}\left(1-\frac{\mu(z)}{\lambda(z)}\right)
	\sqrt{\frac{\lambda(z)}{\mu(z)}\frac{\mu(\iota)}{\lambda(\iota)}} e^{n (V(\iota)-V(z))}
		& \text{if $z < \min\{\iota,\iota'\}$,}\\
	\frac{\lambda(\iota)-\mu(\iota)}{\lambda(\iota)+\mu(\iota)}
		& \text{if $\iota < z < \iota'$,}\\
	\frac{\mu(\iota)-\lambda(\iota)}{\lambda(\iota)+\mu(\iota)}
		& \text{if $\iota' < z < \iota$, and}\\
	\frac{\lambda(\iota)}{\lambda(\iota)+\mu(\iota)}\left(1-\frac{\lambda(z)}{\mu(z)}\right)
	\sqrt{\frac{\mu(z)}{\lambda(z)}\frac{\lambda(\iota)}{\mu(\iota)}} e^{n (V(\iota)-V(z))}
		& \text{if $z > \max\{\iota,\iota'\}$.}
	\end{cases}
\end{multline*}
\end{prop}

\begin{proof}
Since the process can only change by increments of $\pm 1$, for any $i$ and $j$, we have
\begin{align*}
	\mathbb{P}_{i}\left\{T^{(n)}_{j} < T^{(n)}_{i+} \right\} 
	&= \begin{cases}
	\frac{\mu_{i}}{\lambda_{i}+\mu_{i}}\mathbb{P}_{i-1}\left\{T^{(n)}_{j} < T^{(n)}_{i}\right\}
		& \text{if $j < i$, and}\\
	\frac{\lambda_{i}}{\lambda_{i}+\mu_{i}}\mathbb{P}_{i+1}\left\{T^{(n)}_{j} < T^{(n)}_{i}\right\}
		& \text{if $j > i$.}\\
	\end{cases}\\
	&= \begin{cases}
	\frac{\mu_{i}}{\lambda_{i}+\mu_{i}} 
		\frac{e^{n V^{(n)}(i-1)}}{\sum_{k=j}^{i-1} e^{n V^{(n)}(k)}}
		& \text{if $j < i$, and}\\
	\frac{\lambda_{i}}{\lambda_{i}+\mu_{i}} 
		\frac{e^{n V^{(n)}(i)}}{\sum_{k=i}^{j-1} e^{n V^{(n)}(k)}}
			& \text{if $j > i$.}\\
	\end{cases}
\end{align*}
	 
Taking $i = \lfloor \iota n \rfloor$ and $j = \lfloor z n \rfloor$, we are thus left with the task of estimating the sums
\[
	\sum_{k=\lfloor z n \rfloor}^{\lfloor \iota n \rfloor-1} 
		e^{n (V^{(n)}(k)-V^{(n)}(\lfloor \iota n \rfloor -1))}
	\quad \text{and} \quad 
	\sum_{k=\lfloor \iota n \rfloor}^{\lfloor z n \rfloor-1} 
		e^{n (V^{(n)}(k)-V^{(n)}(\lfloor \iota n \rfloor -1))},
\]
using Corollary \ref{LAPLACE2}, where $V(x) - V(\iota)$ finds its maximum at either $z$ or $\iota$ , and this maximum occurs at either the right or left side of the interval of interest ($[z,\iota]$ or $[\iota,z]$), depending on where $z$ lies.  

In particular, if $z < \iota < \kappa < \iota'$ or $z < \iota' < \kappa < \iota$, the interval is $[z,\iota]$,  the maximum occurs at $x = z$ and
\[
	\sum_{k=\lfloor z n \rfloor}^{\lfloor \iota n \rfloor-1} 
		e^{n (V^{(n)}(k)-V^{(n)}(\lfloor \iota n \rfloor -1))} 
		\sim \frac{\sqrt{\frac{\mu(z)}{\lambda(z)}\frac{\lambda(\iota)}{\mu(\iota)}} 
			e^{n (V(z)-V(\iota))}}{1-\frac{\mu(z)}{\lambda(z)}}, 
\]
whereas if $\iota < \kappa < \iota' < z$ or $\iota' < \kappa < \iota <z$, the interval is $[\iota,z]$,  the maximum occurs at $x = z$ and
\[
	\sum_{k=\lfloor \iota n \rfloor}^{\lfloor z n \rfloor-1} 
		e^{n (V^{(n)}(k)-V^{(n)}(\lfloor \iota n \rfloor -1))}
		\sim \frac{\sqrt{\frac{\mu(z)}{\lambda(z)}\frac{\lambda(\iota)}{\mu(\iota)}} 
			e^{n (V(z)-V(\iota))}}{1-\frac{\lambda(z)}{\mu(z)}}.
\]
If if $\iota < z < \iota'$ or $\iota' < z < \iota$, the maximum is at $x = \iota$ whereas the interval is $[\iota,z]$ or $[z,\iota]$ respectively, and one has 
\[
	\sum_{k=\lfloor z n \rfloor}^{\lfloor \iota n \rfloor-1} 
		e^{n (V^{(n)}(k)-V^{(n)}(\lfloor \iota n \rfloor -1))} 
		\sim \frac{1}{1-\frac{\mu(\iota)}{\lambda(\iota)}},
\]
and 
\[
	\sum_{k=\lfloor \iota n \rfloor}^{\lfloor z n \rfloor-1} 
		e^{n (V^{(n)}(k)-V^{(n)}(\lfloor \iota n \rfloor -1))}
		\sim \frac{1}{1-\frac{\lambda(\iota)}{\mu(\iota)}},
\]
respectively.
\end{proof}

As an immediate consequence, we have an estimate of the number of returns to  $\lfloor \iota n \rfloor$:

\begin{prop}\label{RETURNS}
Fix $\iota < \eta$.  Then, 
\[
	\mathbb{E}_{m_{n}}\left[N^{(n)}_{\lfloor \iota n \rfloor}(T^{(n)}_{0})
		\middle\vert T^{(n)}_{\lfloor \iota n \rfloor} < T^{(n)}_{0} \right] 
	\sim \frac{\lambda(\iota)+\mu(\iota)}{\mu(\iota)}
	\sqrt{\frac{\mu(0)}{\lambda(0)}\frac{\lambda(\iota)}{\mu(\iota)}} 
	\frac{e^{-n V(\iota)}}{1-\frac{\mu(0)}{\lambda(0)}}
\]
\end{prop}

\begin{rem}
Since $\iota < \eta$,  $V(\iota) < 0$, and this proposition tells us that as $n \to \infty$, the stochastic logistic process will revisit a neighbourhood of $\lfloor \iota n \rfloor$ an arbitrary number of times before extinction.
\end{rem}

\begin{proof}
Since the process $X^{(n)}(t)$ is Markov, for any integer $i$, each excursion from $i$ is an independent renewal, and thus the number of returns  prior to hitting zero has a geometric distribution with success parameter $\mathbb{P}_{i}\left\{T^{(n)}_{0} < T^{(n)}_{i+}\right\}$:
\[
	\mathbb{P}_{m}\left\{N^{(n)}_{i}(T^{(n)}_{0}) = k \middle\vert T^{(n)}_{\lfloor \iota n \rfloor} < T^{(n)}_{0} \right\}
		= \mathbb{P}_{i}\left\{T^{(n)}_{0} < T^{(n)}_{i+}\right\}
			\left(1-\mathbb{P}_{i}\left\{T^{(n)}_{0} < T^{(n)}_{i+}\right\}\right)^{k-1}
\]
with mean
\[
	\mathbb{E}_{m}\left[N^{(n)}_{i}(T^{(n)}_{0}) \middle\vert T^{(n)}_{\lfloor \iota n \rfloor} < T^{(n)}_{0} \right]
		= \frac{1}{\mathbb{P}_{i}\left\{T^{(n)}_{0} < T^{(n)}_{i+}\right\}}.
\]
The result follows taking $i = \lfloor \iota n \rfloor$, and using the asymptotic for $\mathbb{P}_{\lfloor \iota n \rfloor}\left\{T^{(n)}_{0} < T^{(n)}_{i+}\right\}$ from the previous proposition with $z = 0$, recalling that $V(0)=0$.
\end{proof}

\subsection{Hitting Times}

In this section, we will look at the times to hit 0 and values far above carrying capacity.  Unless stated otherwise, we assume that  $\lambda(x)$ and $\mu(x)$ are twice continuously differentiable.  
We first consider the extinction time for the process, generalising a result that has appeared previously in varying degrees of generality, rigour, and accuracy \cite{Andersson+Djehiche98,Dushoff2000,Newman2004}.

\begin{prop}\label{EXTINCTIONTIME}
Fix a positive integer $m$.  Then,
\[
	\mathbb{E}_{m}\left[T^{(n)}_{0}\right] 
	\sim \sqrt{\frac{2\pi}
		{n\left(\frac{\mu'(\kappa)}{\mu(\kappa)}-\frac{\lambda'(\kappa)}{\lambda(\kappa)}\right)}
		\frac{\mu(0)}{\lambda(0)}}
		\frac{1-\left(\frac{\mu(0)}{\lambda(0)}\right)^{m}}{1-\left(\frac{\mu(0)}{\lambda(0)}\right)}
		\frac{e^{-n V(\kappa)}}
		{\mu(\kappa) \kappa \left(1-\frac{\lambda(0)}{\mu(0)}\right)}
\]
\end{prop}

\begin{proof}
We first observe that the time to extinction is simply the time spent in all states $i > 0$:
\[
	\mathbb{E}_{m}\left[T^{(n)}_{0}\right]
	= \sum_{i = 1}^{\infty} \mathbb{E}_{m}\left[S^{(n)}_{i}(T^{(n)}_{0})\right]\\
	= \sum_{i = 1}^{\infty} \frac{1}{\lambda_{i}+\mu_{i}} \mathbb{E}_{m}\left[N^{(n)}_{i}(T^{(n)}_{0})\right],
\]
as $\frac{1}{\lambda_{i}+\mu_{i}}$ is the expected time spent in state $i$ per visit, and, by definition, 
$N^{(n)}_{i}(T^{(n)}_{0})$ is the total number of visits to $i$ prior to extinction.  As before, $N^{(n)}_{i}(T^{(n)}_{0})$ has a modified geometric distribution with mean
\[
	\frac{\mathbb{P}_{m}\left\{T^{(n)}_{i} < T^{(n)}_{0}\right\}}
		{\mathbb{P}_{i}\left\{T^{(n)}_{0} < T^{(n)}_{i+}\right\}}.
\]

For the numerator, from Proposition \ref{EXTINCTION}, we know that
\[
	\mathbb{P}_{m}\left\{T^{(n)}_{i} < T^{(n)}_{0}\right\} 
		\sim \begin{cases} 1-\left(\frac{\mu(0)}{\lambda(0)}\right)^{m} & \text{if $m < \eta n$, and}\\
		0 & \text{otherwise,}
	\end{cases}
\]
whereas for the denominator, the process can only fail to return to $i$ if the next event is a death and the process hits $0$ prior to hitting $i$:
\[
	\mathbb{P}_{i}\left\{T^{(n)}_{0} < T^{(n)}_{i+}\right\}
	 = \frac{\mu_{i}}{\lambda_{i}+\mu_{i}} \mathbb{P}_{i-1}\left\{T^{(n)}_{0} < T^{(n)}_{i}\right\}
	 = \frac{e^{n V^{(n)}(i-1)}}{\sum_{j=0}^{i-1} e^{n V^{(n)}(j)}}.
\]
Thus, as in the previous result, since $V(0) > V(\frac{i-1}{n})$,
\[
	\sum_{j=0}^{i-1} e^{n V^{(n)}(j)} \sim \frac{1}{1-\frac{\mu(0)}{\lambda(0)}},
\]
and
\[
	\mathbb{P}_{i}\left\{T^{(n)}_{0} < T^{(n)}_{i+}\right\} 
	\sim \left(1-\frac{\mu(0)}{\lambda(0)}\right)\sqrt{\frac{\mu\left(\frac{i-1}{n}\right)}{\mu(0)}
		\frac{\lambda(0)}{\lambda\left(\frac{i-1}{n}\right)}}
	e^{n V\left(\frac{i-1}{n}\right)}.
\]
Thus, since $\mu_{i} = n \mu\left(\frac{i}{n}\right)\frac{i}{n}$,
\[
	\mathbb{E}_{m}\left[T^{(n)}_{0}\right] 
	\sim \frac{1-\left(\frac{\mu(0)}{\lambda(0)}\right)^{m}}{1-\frac{\mu(0)}{\lambda(0)}}
	\sum_{i = 1}^{\lfloor \eta n \rfloor} \frac{e^{-n V\left(\frac{i-1}{n}\right)}}
		{n \mu\left(\frac{i}{n}\right)\frac{i}{n}
	\sqrt{\frac{\mu\left(\frac{i-1}{n}\right)}{\lambda\left(\frac{i-1}{n}\right)}\frac{\lambda(0)}{\mu(0)}}},
\]
and, since $-V(x)$ is maximized at $x = \kappa$, and $\mu(\kappa) = \lambda(\kappa)$, the result follows from Proposition \ref{LAPLACE}.
\end{proof}

Using elements of the proof, we can compute the expected time spent above any level above carrying capacity. Let $L^{(n)}_{i}(t)$ denote the total time spent in state $i$ prior to time $t$ (the \textit{local time} at $i$). Then,

\begin{cor}\label{EXCESSTIME}
Let $\kappa < \iota < \eta$.  Then,
\[
	\mathbb{E}_{m}\left[\sum_{i=\lfloor \iota n \rfloor + 1}^{\infty} L^{(n)}_{i}(T^{(n)}_{0})\right] 
	\sim \frac{1-\left(\frac{\mu(0)}{\lambda(0)}\right)^{m}}{1-\frac{\mu(0)}{\lambda(0)}}
	 \frac{e^{-n V(\iota)}}
		{n \mu(\iota)\iota
	\sqrt{\frac{\mu(\iota)}{\lambda(\iota)}\frac{\lambda(0)}{\mu(0)}}},
\]
\end{cor}

\begin{proof}
The result follows directly from applying Proposition \ref{LAPLACE} to the sum of the expected holding time in all states above $\lfloor \iota n \rfloor$, as computed above.
\end{proof}

We can also completely characterize the amount of time in any given state prior to extinction:

\begin{cor}
Let $m$ be a positive integer.
\begin{multline*}
	\mathbb{P}_{m}\{L^{(n)}_{i}(T^{(n)}_{0}) < t\}\\
		\sim \begin{cases}
		\left(\frac{\mu(0)}{\lambda(0)}\right)^{m} & \text{if $t = 0$, and}\\
		 \left(1-\left(\frac{\mu(0)}{\lambda(0)}\right)^{m}\right) 
		 e^{-(\lambda_{i} + \mu_{i})\left(1-\left(1-\frac{\mu(0)}{\lambda(0)}\right)\sqrt{\frac{\mu\left(\frac{i-1}{n}\right)}{\mu(0)}
		\frac{\lambda(0)}{\lambda\left(\frac{i-1}{n}\right)}}
	e^{n V\left(\frac{i-1}{n}\right)}\right) t} 
		 & \text{otherwise}
		\end{cases}
\end{multline*}
\ie, conditioned upon hitting $i$ prior to $0$, the total time spent in $i$ is exponentially distributed with rate $(\lambda_{i} + \mu_{i})\left(1-\left(1-\frac{\mu(0)}{\lambda(0)}\right)\sqrt{\frac{\mu\left(\frac{i-1}{n}\right)}{\mu(0)}\frac{\lambda(0)}{\lambda\left(\frac{i-1}{n}\right)}} e^{n V\left(\frac{i-1}{n}\right)}\right)$.  	
\end{cor}

\begin{proof}
By the strong Markov property, each excursion starting from state $i$ is independent.  Thus, conditional on $n_{i}$ visits to $i$, the time spent in $i$ after each return is a sum of $n_{i}$ independent exponentially distributed random variables with rate $\lambda_{i} + \mu_{i}$ \ie a gamma-distributed with shape and rate parameters $n_{i}$ and $\lambda_{i} + \mu_{i}$: the probability that the total time is in $[t,t+dt)$ is 
\begin{multline*}
	\int_{0}^{t}\int_{0}^{t-t_{1}} \cdots \int_{0}^{t-t_{1}-t_{2}-\cdots-t_{n_{i}-2}}
		\prod_{j=1}^{n_{i}-1} (\lambda_{i} + \mu_{i}) e^{-(\lambda_{i} + \mu_{i})t_{j}}\\
		 \times (\lambda_{i} + \mu_{i}) e^{-(\lambda_{i} + \mu_{i})(t-t_{1}-t_{2}-\cdots-t_{n_{i}-1})}
		 \, dt_{1}dt_{2}\cdots dt_{n_{i}-1}\\
		 = \frac{(\lambda_{i} + \mu_{i})^{n_{i}}}{(n_{i}-1)!} t^{n_{i}-1} e^{-(\lambda_{i} + \mu_{i})t}.
\end{multline*}
Now, we observed above that the number of visits to $i$ prior to extinction, $N^{(n)}_{i}(T^{(n)}_{0})$, has a modified geometric distribution, with probability $\mathbb{P}_{m}\left\{T^{(n)}_{i} < T^{(n)}_{0}\right\}$ of reaching $i$ prior to extinction, and return probability $\mathbb{P}_{i}\left\{T^{(n)}_{0} < T^{(n)}_{i+}\right\}$.  The former gives the probability that $L^{(n)}_{i}(T^{(n)}_{0}) > 0$, whereas summing over the distribution of $N^{(n)}_{i}(T^{(n)}_{0})$ the probability that $L^{(n)}_{i}(T^{(n)}_{0}) \in [t,t+dt)$ is
\begin{multline*}
	\mathbb{P}_{m}\left\{T^{(n)}_{i} < T^{(n)}_{0}\right\}
	\left(1-\mathbb{P}_{i}\left\{T^{(n)}_{0} < T^{(n)}_{i+}\right\}\right)\\
	\times \sum_{n_{i} = 1}^{\infty}  \frac{(\lambda_{i} + \mu_{i})^{n_{i}}}{(n_{i}-1)!} t^{n_{i}-1} 
		e^{-(\lambda_{i} + \mu_{i})t}\mathbb{P}_{i}\left\{T^{(n)}_{0} < T^{(n)}_{i+}\right\}^{n_{i}-1}\\
		= \mathbb{P}_{m}\left\{T^{(n)}_{i} < T^{(n)}_{0}\right\} (\lambda_{i} + \mu_{i}) 
		\left(1-\mathbb{P}_{i}\left\{T^{(n)}_{0} < T^{(n)}_{i+}\right\}\right)
		e^{-(\lambda_{i} + \mu_{i}) \left(1-\mathbb{P}_{i}\left\{T^{(n)}_{0} < T^{(n)}_{i+}\right\}\right)t}.
\end{multline*}
The result then follows using the asymptotic approximations of the previous proof. 
\end{proof}

The fact that $\mathbb{E}_{m}\left[T^{(n)}_{0}\right]$ is bounded in $m$ also allows us to show that the distribution of $T^{(n)}_{0}$ has exponential tails:

\begin{cor}\label{ALDOUS}
Independent of the initial state, $m$,
\[
	\mathbb{P}_{m}\left\{T^{(n)}_{0} > t\right\} \lesssim 
	e^{-\frac{t}{e} 
	\sqrt{\frac{n\left(\frac{\mu'(\kappa)}{\mu(\kappa)}-\frac{\lambda'(\kappa)}{\lambda(\kappa)}\right)}{2\pi}\frac{\lambda(0)}{\mu(0)}}\mu(\kappa) \kappa \left(1-\frac{\lambda(0)}{\mu(0)}\right) e^{n V(\kappa)}}
\]
\end{cor}

\begin{proof}
This is a special case of a result in Section 4.3 of \cite{Aldous2002}, but the proof is short, so we give it for completeness.   Let $\theta > 0$ be an arbitrary constant.  Then, using the Markov property of $X^{(n)}(t)$ and Markov's inequality,
\begin{align*}
	\mathbb{P}_{m}\left\{T^{(n)}_{0} > j\theta \middle\vert T^{(n)}_{0} > (j-1)\theta\right\}
	&= \mathbb{E}_{m}\left[\mathbb{P}_{X^{(n)}((j-1)\theta)}\left\{T^{(n)}_{0} > \theta\right\}\right]\\
	&\leq \sup_{i > 0} \mathbb{P}_{i}\left\{T^{(n)}_{0} > \theta\right\}\\
	&\leq \sup_{i > 0} \frac{\mathbb{E}_{i}\left[T^{(n)}_{0}\right]}{\theta}.
\end{align*}
Thus, by induction,
\[
	\mathbb{P}_{m}\left\{T^{(n)}_{0} > j\theta\right\} 
	\leq \left(\frac{\sup_{i > 0} \mathbb{E}_{i}\left[T^{(n)}_{0}\right]}{\theta}\right)^{j},
\]
\ie 
\[
	\mathbb{P}_{m}\left\{T^{(n)}_{0} > t \right\} 
	\leq e^{\frac{t}{\theta}\ln\left(\frac{\sup_{i > 0} \mathbb{E}_{i}\left[T^{(n)}_{0}\right]}{\theta}	
	\right)}.
\]
This is minimized by taking $\theta = e \sup_{i > 0} \mathbb{E}_{i}\left[T^{(n)}_{0}\right]$.   Finally,  Propostion \ref{EXTINCTIONTIME} shows that 
\[
	 \sup_{i > 0} \mathbb{E}_{i}\left[T^{(n)}_{0}\right]
	 \sim \sqrt{\frac{2\pi}
		{n\left(\frac{\mu'(\kappa)}{\mu(\kappa)}-\frac{\lambda'(\kappa)}{\lambda(\kappa)}\right)}
		\frac{\mu(0)}{\lambda(0)}} \frac{e^{-n V(\kappa)}}
		{\mu(\kappa) \kappa \left(1-\frac{\lambda(0)}{\mu(0)}\right)}.
\]
\end{proof}

We next find the expected time to hit the carrying capacity, conditioned upon reaching it before extinction:

\begin{prop}\label{CAPACITYTIME}
Fix a positive integer $m$.  Then,
\[
	\mathbb{E}_{m}\left[T^{(n)}_{\lfloor \kappa n\rfloor} 
		\middle\vert T^{(n)}_{\lfloor \kappa n\rfloor}< T^{(n)}_{0}\right] \sim 
		\left(\frac{1}{\lambda(0)-\mu(0)} - \frac{1}{(\lambda'(\kappa)-\mu'(\kappa))\kappa}\right)
		\ln{n}.
\]
\end{prop}

\begin{proof}
The proof presented here is based upon the treatment given for the Moran model in \cite{Durrett2009}.  We first observe that the hitting time of $0$ or $\lfloor \kappa n\rfloor$ is the sum of the time spent in all in-between states prior to $T^{(n)}_{\lfloor \kappa n\rfloor} $, so that 
\begin{multline*}
	\mathbb{E}_{m}\left[T^{(n)}_{\lfloor \kappa n\rfloor} 
		\middle\vert T^{(n)}_{\lfloor \kappa n\rfloor} < T^{(n)}_{0}\right]
	= \sum_{i = 1}^{\lfloor \kappa n\rfloor - 1} 
		\mathbb{E}_{m}\left[S^{(n)}_{i}(T^{(n)}_{\lfloor \kappa n\rfloor} )
			\middle\vert T^{(n)}_{\lfloor \kappa n\rfloor} < T^{(n)}_{0}\right]\\
%
		= \frac{1}{\lambda_{i}+\mu_{i}} 
		\mathbb{E}_{m}\left[N^{(n)}_{i}(T^{(n)}_{\lfloor \kappa n\rfloor} )
		\middle\vert T^{(n)}_{\lfloor \kappa n\rfloor} < T^{(n)}_{0}\right],
\end{multline*}
as $\frac{1}{\lambda_{i}+\mu_{i}}$ is the expected time spent in state $i$, which is exponentially distributed with parameter $\lambda_{i}+\mu_{i}$.

Now, $N^{(n)}_{m}(T^{(n)}_{\lfloor \kappa n\rfloor} )$ has a modified geometric distribution:
\begin{multline*}
	\mathbb{P}_{m}\left\{N^{(n)}_{i}(T^{(n)}_{\lfloor \kappa n\rfloor} ) = k 
		\middle\vert T^{(n)}_{\lfloor \kappa n\rfloor} < T^{(n)}_{0}\right\}\\ = \begin{cases}
		\mathbb{P}_{m}\left\{T^{(n)}_{\lfloor \kappa n\rfloor} < T^{(n)}_{i}
			\middle\vert T^{(n)}_{\lfloor \kappa n\rfloor} < T^{(n)}_{0} \right\}  
			& \text{if $k = 0$, and}\\
	\begin{multlined}
		\mathbb{P}_{m}\left\{T^{(n)}_{i} < T^{(n)}_{\lfloor \kappa n\rfloor} 
			\middle\vert T^{(n)}_{\lfloor \kappa n\rfloor} < T^{(n)}_{0} \right\}\\
		\times \mathbb{P}_{i}\left\{T^{(n)}_{\lfloor \kappa n\rfloor} < T^{(n)}_{i+}
			\middle\vert T^{(n)}_{\lfloor \kappa n\rfloor} < T^{(n)}_{0}\right\}\\
		\times \left(1-\mathbb{P}_{i}\left\{T^{(n)}_{\lfloor \kappa n\rfloor} < T^{(n)}_{i+}
			\middle\vert T^{(n)}_{\lfloor \kappa n\rfloor} < T^{(n)}_{0}\right\}\right)^{k-1}
	\end{multlined}	
			& \text{if $k\geq 1$,}
		\end{cases}
\end{multline*}
which has mean
\[
	\frac{\mathbb{P}_{m}\left\{T^{(n)}_{i} < T^{(n)}_{\lfloor \kappa n\rfloor} 
		\middle\vert T^{(n)}_{\lfloor \kappa n\rfloor} < T^{(n)}_{0}\right\}}
		{\mathbb{P}_{i}\left\{T^{(n)}_{\lfloor \kappa n\rfloor} < T^{(n)}_{i+}
			\middle\vert T^{(n)}_{\lfloor \kappa n\rfloor} < T^{(n)}_{0}\right\}}.
\]
Now, if we specialize to the case when $m = 1$, then the process must pass through $i$ en route to $\lfloor \kappa n\rfloor$,
so 
\[
	\mathbb{P}_{1}\left\{T^{(n)}_{i} < T^{(n)}_{\lfloor \kappa n\rfloor} 
		\middle\vert T^{(n)}_{\lfloor \kappa n\rfloor} < T^{(n)}_{0} \right\} = 1
\]

Moreover, conditional on $T^{(n)}_{\lfloor \kappa n\rfloor} < T^{(n)}_{0}$, starting from $i$, $T^{(n)}_{\lfloor \kappa n\rfloor} < T^{(n)}_{i+}$ if and only if a birth occurs \textit{and} the process hits $\lfloor \kappa n\rfloor$ prior to $i$:
\[
	\mathbb{P}_{i}\left\{T^{(n)}_{\lfloor \kappa n\rfloor} < T^{(n)}_{i+} 
		\middle\vert T^{(n)}_{\lfloor \kappa n\rfloor} < T^{(n)}_{0}\right\}
	= \frac{\lambda_{i}}{\lambda_{i}+\mu_{i}} 
		\mathbb{P}_{i+1}\left\{T^{(n)}_{\lfloor \kappa n\rfloor} < T^{(n)}_{i}
			\middle\vert T^{(n)}_{\lfloor \kappa n\rfloor} < T^{(n)}_{0}\right\}
\]
so that 
\[	
	\mathbb{E}_{1}\left[T^{(n)}_{\lfloor \kappa n\rfloor} 
		\middle\vert T^{(n)}_{\lfloor \kappa n\rfloor}< T^{(n)}_{0}\right] = 
	\sum_{i = 1}^{\lfloor \kappa n\rfloor - 1} \frac{1}
	{\lambda_{i}\mathbb{P}_{i+1}\left\{T^{(n)}_{\lfloor \kappa n\rfloor} < T^{(n)}_{i}
		\middle\vert T^{(n)}_{\lfloor \kappa n\rfloor} < T^{(n)}_{0}\right\}}
\]
whereas
\[
	\mathbb{P}_{i+1}\left\{T^{(n)}_{\lfloor \kappa n\rfloor} < T^{(n)}_{i}
		\middle\vert T^{(n)}_{\lfloor \kappa n\rfloor} < T^{(n)}_{0}\right\}
	= \frac{\mathbb{P}_{i+1}\left\{T^{(n)}_{\lfloor \kappa n\rfloor} < T^{(n)}_{i}\right\}}
		{\mathbb{P}_{i+1}\left\{T^{(n)}_{\lfloor \kappa n\rfloor} < T^{(n)}_{0}\right\}},
\]
since $\left\{T^{(n)}_{\lfloor \kappa n\rfloor} < T^{(n)}_{i}\right\} \subseteq \left\{T^{(n)}_{\lfloor \kappa n\rfloor} < T^{(n)}_{0}\right\}$, as to reach 0 from $i+1$, the process must pass via $i$.

Now, Proposition \ref{EXTINCTION} and its proof tell us that $\mathbb{P}_{i+1}\left\{T^{(n)}_{\lfloor \kappa n\rfloor} < T^{(n)}_{0}\right\}$ tends to 1 as $i \to \infty$, and, moreover, that this convergence is uniform in $n$. We may thus apply Lemma \ref{DROP}, to conclude that 
\[
	\mathbb{E}_{1}\left[T^{(n)}_{\lfloor \kappa n\rfloor} 
		\middle\vert T^{(n)}_{\lfloor \kappa n\rfloor}< T^{(n)}_{0}\right] \sim
		\sum_{i = 1}^{\lfloor \kappa n\rfloor - 1} \frac{1}
	{\lambda_{i}\mathbb{P}_{i+1}\left\{T^{(n)}_{\lfloor \kappa n\rfloor} < T^{(n)}_{i}\right\}}
\]		

We now observe that
\[
	\mathbb{P}_{i+1}\left\{T^{(n)}_{\lfloor \kappa n\rfloor} < T^{(n)}_{i}\right\}
	= h^{(n)}_{\lfloor \kappa n\rfloor,i}(i+1)
	 	= \frac{e^{nV^{(n)}(i)}}{\sum_{j = i}^{\lfloor \kappa n\rfloor-1} e^{nV^{(n)}(j)}},
\]
which, by Corollary \ref{LAPLACE2} is asymptotically equivalent to $1-\frac{\mu\left(\frac{i}{n}\right)}{\lambda\left(\frac{i}{n}\right)}$, so recalling that $\lambda_{i} = \lambda\left(\frac{i}{n}\right) i$,
\[
	 \sum_{i = 1}^{\lfloor \kappa n\rfloor - 1} \frac{1}
	{\lambda_{i}\mathbb{P}_{i+1}\left\{T^{(n)}_{\lfloor \kappa n\rfloor} < T^{(n)}_{i}\right\}}
	\sim \sum_{i = 1}^{\lfloor \kappa n\rfloor - 1} \frac{1}{\left(\lambda\left(\frac{i}{n}\right) - 
		\mu\left(\frac{i}{n}\right)\right) i}.
\]
The latter is the Riemann sum for the integral of $\frac{1}{(\lambda(x)-\mu(x))x}$ over $[0,\kappa]$, but this integral diverges at both endpoints.  To deal with this, first observe that, using Taylor's theorem, we may write 
\[
	\lambda(x)-\mu(x) = \lambda(0)-\mu(0) +  (\lambda'(0)-\mu'(0) + r(x))x,
\]
where $r(x) \to 0$ as $x \to 0$ and
\[
	\lambda(x)-\mu(x) =  (\lambda'(\kappa)-\mu'(\kappa))(x-\kappa)
		 +(\lambda''(\kappa)-\mu''(\kappa)+R(x))(x-\kappa)^{2},
\]
for a continuous function $R(x)$ such that  $R(x) \to 0$ as $x \to \kappa$.  Then, for arbitrary $\varepsilon > 0$, we can choose $n$ sufficiently large that 
\[
	{\textstyle \lambda(0)-\mu(0) < \lambda\left(\frac{i}{n}\right)-\mu\left(\frac{i}{n}\right) 
		< \lambda(0)-\mu(0) + \varepsilon}
\]
for all $i \leq \frac{n}{\ln{n}}$ and
\[
	{\textstyle  (\lambda'(\kappa)-\mu'(\kappa))\left(\frac{i}{n}-\kappa\right) - \varepsilon 
		< \lambda\left(\frac{i}{n}\right)-\mu\left(\frac{i}{n}\right) 
		< (\lambda'(\kappa)-\mu'(\kappa))\left(\frac{i}{n}-\kappa\right) + \varepsilon}
\]
for all $\lfloor \kappa n\rfloor - \left\lfloor \frac{n}{\ln{n}} \right\rfloor  \leq i < \lfloor \kappa n\rfloor $, and split the sum in three:
\[
	\sum_{i = 1}^{\left\lfloor \frac{n}{\ln{n}} \right\rfloor } 
		\frac{1}{\left(\lambda\left(\frac{i}{n}\right) - \mu\left(\frac{i}{n}\right)\right) i} +
	\sum_{i = \left\lfloor \frac{n}{\ln{n}} \right\rfloor +1}^{\lfloor \kappa n\rfloor - \left\lfloor \frac{n}{\ln{n}} \right\rfloor - 1} 
		\frac{1}{\left(\lambda\left(\frac{i}{n}\right) - \mu\left(\frac{i}{n}\right)\right) i} +
	\sum_{i = \lfloor \kappa n\rfloor - \left\lfloor \frac{n}{\ln{n}} \right\rfloor }^{\lfloor \kappa n\rfloor - 1} 
		\frac{1}{\left(\lambda\left(\frac{i}{n}\right) - \mu\left(\frac{i}{n}\right)\right) i}.
\]

For the first sum, we have that 
\[
	\frac{1}{\lambda(0)-\mu(0) + \varepsilon} \sum_{i = 1}^{\left\lfloor \frac{n}{\ln{n}} \right\rfloor } \frac{1}{i}
	\leq \sum_{i = 1}^{\left\lfloor \frac{n}{\ln{n}} \right\rfloor } 
		\frac{1}{\left(\lambda\left(\frac{i}{n}\right) - \mu\left(\frac{i}{n}\right)\right) i} 
	\leq \frac{1}{\lambda(0)-\mu(0)} \sum_{i = 1}^{\left\lfloor \frac{n}{\ln{n}} \right\rfloor } \frac{1}{i},
\]
whereas 
\[
	\sum_{i = 1}^{\left\lfloor \frac{n}{\ln{n}} \right\rfloor } \frac{1}{i} = \ln{\left\lfloor \frac{n}{\ln{n}} \right\rfloor} 
	+ \gamma + \epsilon_{\left\lfloor \frac{n}{\ln{n}} \right\rfloor },
\]
where $\gamma$ is the Euler-Mascheroni constant and $\epsilon_{\left\lfloor \frac{n}{\ln{n}} \right\rfloor } \sim \frac{1}{2\left\lfloor \frac{n}{\ln{n}} \right\rfloor }$.

Similarly,	
\begin{multline*}
	\frac{1}{(\lambda'(\kappa)-\mu'(\kappa)) + \varepsilon)\kappa}
	\sum_{i = \lfloor \kappa n\rfloor - \left\lfloor \frac{n}{\ln{n}} \right\rfloor }^{\lfloor \kappa n\rfloor - 1} 
		\frac{1}{i-\kappa n}
	\leq
	\sum_{i =\lfloor \kappa n\rfloor - \left\lfloor \frac{n}{\ln{n}} \right\rfloor }^{\lfloor \kappa n\rfloor - 1} 
		\frac{1}{\left(\lambda\left(\frac{i}{n}\right) - \mu\left(\frac{i}{n}\right)\right) i}\\
	\leq \frac{1}{(\lambda'(\kappa)-\mu'(\kappa)) - \varepsilon)(\kappa - \delta)} 
	\sum_{i = \lfloor \kappa n\rfloor - \left\lfloor \frac{n}{\ln{n}} \right\rfloor }^{\lfloor \kappa n\rfloor - 1} \frac{1}{i-\kappa n}
\end{multline*}
and, 
\[
	\sum_{i = \lfloor \kappa n\rfloor - \left\lfloor \frac{n}{\ln{n}} \right\rfloor }^{\lfloor \kappa n\rfloor - 1} 
		\frac{1}{i-\kappa n}
	= - \sum_{i = 1}^{\left\lfloor \frac{n}{\ln{n}} \right\rfloor } 
		\frac{1}{i+\kappa n-\lfloor \kappa n\rfloor}.
\]
Since $0 \leq \kappa n - \lfloor \kappa n\rfloor < 1$,
\[
	\sum_{i = 1}^{\left\lfloor \frac{n}{\ln{n}} \right\rfloor } \frac{1}{i+1}
	<   \sum_{i = 1}^{\left\lfloor \frac{n}{\ln{n}} \right\rfloor } 
		\frac{1}{i+\kappa n-\lfloor \kappa n\rfloor}
	\leq \sum_{i = 1}^{\left\lfloor \frac{n}{\ln{n}} \right\rfloor } \frac{1}{i} 
\]

Finally, to deal with the middle sum, we first note that 
\[
	\frac{d}{dx} \frac{1}{(\lambda(x)-\mu(x))x} = 
	-\frac{(\lambda'(x)-\mu'(x))x + \lambda(x)-\mu(x)}{(\lambda(x)-\mu(x))^{2}x^{2}}
\]
is bounded on any closed interval in $(0,\kappa)$ and tends to $+\infty$ at $0$, where it is decreasing, and at $\kappa$, where it is increasing; in particular, on $\left[\frac{1}{\ln{n}},\kappa - \frac{1}{\ln{n}}\right]$ the derivative is bounded above by its values at the endpoints, which are bounded above by 
\[
	\frac{\sup_{x \in [0,\kappa]} -(\lambda'(x)-\mu'(x))x}
	{\min\{\lambda(0)-\mu(0),(\lambda'(\kappa)-\mu'(\kappa))\kappa\}} (\ln{n})^{2}.
\]
Thus, applying Lemma \ref{INEQ}, we have that 
\begin{multline*}
	\left|
	\sum_{i = \left\lfloor \frac{n}{\ln{n}} \right\rfloor +1}^{\lfloor \kappa n\rfloor - \left\lfloor \frac{n}{\ln{n}} \right\rfloor - 1} 
		\frac{1}{\left(\lambda\left(\frac{i}{n}\right) - \mu\left(\frac{i}{n}\right)\right) i} 
	- \int_{\frac{1}{n}\left(\left\lfloor \frac{n}{\ln{n}} \right\rfloor +1\right)}
		^{\frac{1}{n}\left(\lfloor \kappa n\rfloor - \left\lfloor \frac{n}{\ln{n}} \right\rfloor \right)}
		\frac{dx}{(\lambda(x)-\mu(x))x}\right|\\
		\leq \frac{\sup_{x \in [0,\kappa]} -(\lambda'(x)-\mu'(x))x}
	{\min\{\lambda(0)-\mu(0),(\lambda'(\kappa)-\mu'(\kappa))\kappa\}} \frac{(\ln{n})^{2}}{2n}
\end{multline*}
Moreover,
\[
	 0 \leq \int_{\frac{1}{n}\left(\left\lfloor \frac{n}{\ln{n}} \right\rfloor +1\right)}
		^{\frac{1}{n}\left(\lfloor \kappa n\rfloor - \left\lfloor \frac{n}{\ln{n}} \right\rfloor \right)}
		\frac{dx}{(\lambda(x)-\mu(x))x}
	\leq  \int_{\frac{1}{\ln{n}}}^{ \kappa - \frac{1}{\ln{n}}} \frac{dx}{(\lambda(x)-\mu(x))x},
\]
and, since $r(x)$ and $R(x)$ are continuous, and thus bounded on $[0,\kappa]$,
\[
	\int_{\frac{1}{\ln{n}}}^{\frac{\kappa}{2}} \frac{dx}{(\lambda(x)-\mu(x))x}
		- \int_{\frac{1}{\ln{n}}}^{\frac{\kappa}{2}} \frac{dx}{(\lambda(0)-\mu(0))x} 
	= \int_{\frac{1}{\ln{n}}}^{\frac{\kappa}{2}} \frac{\lambda'(0)-\mu'(0) + r(x)}{
		(\lambda(0)-\mu(0))(\lambda(x)-\mu(x))}\,dx 
\]
and
\begin{multline*}
	\int_{\frac{\kappa}{2}}^{\kappa-\frac{1}{\ln{n}}} \frac{dx}{(\lambda(x)-\mu(x))x}
		- \int_{\frac{\kappa}{2}}^{\kappa-\frac{1}{\ln{n}}} 
			\frac{dx}{(\lambda'(\kappa)-\mu'(\kappa))\kappa(x-\kappa)} \\
	= \int_{\frac{1}{\ln{n}}}^{\frac{\kappa}{2}} \frac{\lambda''(\kappa)-\mu''(\kappa)+R(x)}{
		(\lambda'(\kappa)-\mu'(\kappa))\kappa h(x)}\,dx 
\end{multline*}
are bounded, where
\[
	h(x) = \begin{cases} 
		\frac{(\lambda(x)-\mu(x))x}{(x-\kappa)} & \text{for $x \neq \kappa$, and}\\
		(\lambda'(\kappa)-\mu'(\kappa))\kappa & \text{for $x = \kappa$.}
	\end{cases}
\]

Finally, we observe that 
\[
	\int_{\frac{1}{\ln{n}}}^{\frac{\kappa}{2}} \frac{dx}{(\lambda(0)-\mu(0))x}
	+ \frac{1}{\lambda(0)-\mu(0)}\left(\ln{\frac{\kappa}{2}}-\ln{\left(\frac{1}{\ln{n}}\right)}\right) 
\]
and
\[
	\int_{\frac{\kappa}{2}}^{\kappa-\frac{1}{\ln{n}}} 
		\frac{dx}{(\lambda'(\kappa)-\mu'(\kappa))\kappa(x-\kappa)}\\ 
		= \frac{1}{(\lambda'(\kappa)-\mu'(\kappa))\kappa}
		\left(\ln{\left(\frac{1}{\ln{n}}\right)}-\ln{\frac{\kappa}{2}}\right),
\]
so that the middle sum is $\BigO{\ln{\ln{n}}}$.

Since the choice of $\varepsilon$ is arbitrary, the result follows.
\end{proof}


We next consider the time for a fluctuation far above carrying capacity.     

\begin{prop}\label{EXCURSIONTIME}
Let $\iota < \eta$.  Then, for fixed $m \in \mathbb{N}$,
\[
	 \mathbb{E}_{m}\left[T^{(n)}_{\lfloor \iota n \rfloor} 
		\middle\vert T^{(n)}_{\lfloor \iota n \rfloor} < T^{(n)}_{0}\right] \sim
		\sqrt{\frac{2\pi}
		{n\left(\frac{\mu'(\kappa)}{\mu(\kappa)}-\frac{\lambda'(\kappa)}{\lambda(\kappa)}\right)}
		\frac{\mu(\iota)}{\lambda(\iota)}}
		\frac{e^{n(V(\iota)-V(\kappa))}}{\lambda(\kappa) \kappa 
			\left(1-\frac{\lambda(\iota)}{\mu(\iota)}\right)}
\]
\end{prop}

 
\begin{proof}
Proceeding as previously, we have that
\[
	\mathbb{E}_{1}\left[T^{(n)}_{\lfloor \iota n \rfloor} 
		\middle\vert T^{(n)}_{\lfloor \iota n \rfloor} < T^{(n)}_{0}\right] 
		\sim \sum_{i = 1}^{\lfloor \iota n\rfloor - 1} \frac{1}
	{\lambda_{i}\mathbb{P}_{i+1}\left\{T^{(n)}_{\lfloor \iota n\rfloor} < T^{(n)}_{i}\right\}}
\]
and
\[
	\mathbb{P}_{i+1}\left\{T^{(n)}_{\lfloor \iota n\rfloor} < T^{(n)}_{i}\right\}
	 	= \frac{1}{\sum_{j = i}^{\lfloor \iota n\rfloor-1} e^{n(V^{(n)}(j)-V^{(n)}(i))}}.
\]

Now, given that $V(x)$ is concave with a minimum at $\kappa$, and $V(0) = 0$, there exists $0 < \iota' < \kappa$ such that $V(\iota') = V(\iota)$.  Then, for $i < \lfloor \iota' n\rfloor$, $V\left(\frac{j}{n}\right) - V\left(\frac{i}{n}\right)$ is maximized at $j = i$, whereas for $\lfloor \iota' n\rfloor < i < \lfloor \iota n\rfloor$, it is maximized at $j = \lfloor \iota n\rfloor - 1$.

We thus have
\[
	\mathbb{P}_{i+1}\left\{T^{(n)}_{\lfloor \iota n\rfloor} < T^{(n)}_{i}\right\}
	\sim 1-\frac{\mu\left(\frac{i}{n}\right)}{\lambda\left(\frac{i}{n}\right)}
\]
for $i < \lfloor \iota' n\rfloor$, whereas for $\lfloor \iota' n\rfloor < i < \lfloor \iota n\rfloor$,
\[
	\frac{1}{\mathbb{P}_{i+1}\left\{T^{(n)}_{\lfloor \iota n\rfloor} < T^{(n)}_{i}\right\}}
	\sim \frac{\sqrt{\frac{\mu(\iota)\lambda\left(\frac{i}{n}\right)}{\lambda(\iota)\mu\left(\frac{i}{n}\right)}}
		e^{n\left(V(\iota) - V\left(\frac{i}{n}\right)\right)}}{1-\frac{\lambda(\iota)}{\mu(\iota)}}
\]

We now split the sum over $i$ at $\lfloor \iota' n \rfloor$.  Then,
\[
 	\sum_{i = 1}^{\lfloor \iota' n\rfloor - 1} \frac{1}
	{\lambda_{i}\mathbb{P}_{i+1}\left\{T^{(n)}_{\lfloor \iota n\rfloor} < T^{(n)}_{i}\right\}}
	\sim \sum_{i = 1}^{\lfloor \iota' n\rfloor - 1} 
	\frac{1}{\left(\lambda\left(\frac{i}{n}\right)-\mu\left(\frac{i}{n}\right)\right)i},
\]
whereas
\[
	\lambda(0)-\mu(0) \leq \lambda\left(\frac{i}{n}\right)-\mu\left(\frac{i}{n}\right)
	\leq \lambda(\iota)-\mu(\iota),
\]
so, as previously, 
\begin{multline*}
	\frac{1}{\lambda(\iota)-\mu(\iota)} 
	\leq \liminf_{n \to \infty} \frac{1}{\ln (\iota' n)} \sum_{i = 1}^{\lfloor \iota' n\rfloor - 1} \frac{1}
		{\lambda_{i}\mathbb{P}_{i+1}\left\{T^{(n)}_{\lfloor \iota n\rfloor} < T^{(n)}_{i}\right\}}\\
	\leq \limsup_{n \to \infty} \frac{1}{\ln (\iota' n)} \sum_{i = 1}^{\lfloor \iota' n\rfloor - 1} \frac{1}
		{\lambda_{i}\mathbb{P}_{i+1}\left\{T^{(n)}_{\lfloor \iota n\rfloor} < T^{(n)}_{i}\right\}}
	\leq \frac{1}{\lambda(0)-\mu(0)}.
\end{multline*}

On the other hand, we observe that for $x \in [\iota',\iota]$, $V(\iota) - V(x)$ is maximized at
$x = \kappa$, so that, applying Proposition \ref{LAPLACE}, we have 
\begin{multline*}
 	\sum_{i = \lfloor \iota' n\rfloor}^{\lfloor \iota n\rfloor - 1} \frac{1}
	{\lambda_{i}\mathbb{P}_{i+1}\left\{T^{(n)}_{\lfloor \iota n\rfloor} < T^{(n)}_{i}\right\}}
	\sim 
 	\sum_{i = \lfloor \iota' n\rfloor}^{\lfloor \iota n\rfloor - 1}
 	\frac{e^{n\left(V(\iota) - V\left(\frac{i}{n}\right)\right)}} 
	{n \lambda\left(\frac{i}{n}\right)\left(\frac{i}{n}\right)\left(1-\frac{\lambda(\iota)}{\mu(\iota)}\right)}\\
	\sim \sqrt{\frac{2\pi}
		{n\left(\frac{\mu'(\kappa)}{\mu(\kappa)}-\frac{\lambda'(\kappa)}{\lambda(\kappa)}\right)}}
		\frac{\sqrt{\frac{\mu(\iota)}{\lambda(\iota)}\frac{\lambda(\kappa)}{\mu(\kappa)}}
		e^{n(V(\iota)-V(\kappa))}}{\lambda(\kappa) \kappa 
		\left(1-\frac{\lambda(\iota)}{\mu(\iota)}\right)}.
\end{multline*}
The result follows on observing that $\lambda(\kappa) = \mu(\kappa)$.


\end{proof}

\begin{rem}
The proof of Corollary \ref{ALDOUS}, suitably adapted, also shows that for all $m$,
\[
	 \mathbb{P}_{m}\left\{T^{(n)}_{\lfloor \iota n \rfloor} > t
	\middle\vert T^{(n)}_{\lfloor \iota n \rfloor} < T^{(n)}_{0}\right\} \lesssim e^{-\frac{t}{e}
	\sqrt{
	\frac{n\left(\frac{\mu'(\kappa)}{\mu(\kappa)}-\frac{\lambda'(\kappa)}{\lambda(\kappa)}\right)}
	{2\pi} \frac{\lambda(\iota)}{\mu(\iota)}}
	\lambda(\kappa) \kappa \left(1-\frac{\lambda(\iota)}{\mu(\iota)}\right)
	e^{n(V(\kappa)-V(\iota))}}
\]
\end{rem} 

\begin{rem}
We note that in the proof, the time spent in states $i = 1,\ldots,\lfloor \iota' n\rfloor-1$ is extremey short ($\BigO{\ln n}$) compared to the time in states $i = \lfloor \iota' n\rfloor,\ldots,\lfloor \iota n\rfloor-1$; this is because the process rapidly approaches the carrying capacity, but, prior to hitting 
$\lfloor \iota n\rfloor$, it makes many returns to each point such that $V\left(\frac{i}{n}\right) < V(\iota)$.
\end{rem}

Moreover, this return is quite rapid:

\begin{prop}\label{RETURNTIME}
Let $\kappa < \iota < \eta$.  Then, 
\[
	\lim_{n \to \infty} \mathbb{E}_{\lfloor \iota n \rfloor}\left[T^{(n)}_{\lfloor \kappa n \rfloor}\right]
	\sim - \frac{1}{(\lambda'(\kappa)-\mu'(\kappa))\kappa}\ln{n}
\]
\end{prop}

\begin{proof}
As previously, we have that 
\[
	\mathbb{E}_{\lfloor \iota n \rfloor}\left[T^{(n)}_{\lfloor \kappa n \rfloor}\right]
	= \sum_{i = \lfloor \kappa n \rfloor + 1}^{\infty} \frac{1}{n (\lambda_{i}+\mu_{i})}
	\frac{\mathbb{P}_{\lfloor \iota n \rfloor} \left\{T^{(n)}_{i} < T^{(n)}_{\lfloor \kappa n \rfloor}\right\}}
	{\mathbb{P}_{i}\left\{T^{(n)}_{\lfloor \kappa n \rfloor} < T^{(n)}_{i+}\right\}}. 
\]
Now, 
\[
	\mathbb{P}_{\lfloor \iota n \rfloor} \left\{T^{(n)}_{i} < T^{(n)}_{\lfloor \kappa n \rfloor}\right\} 
	= \begin{cases}
	1 & \text{if $\lfloor \kappa n \rfloor < i \leq \lfloor \iota n \rfloor$, and}\\
	\frac{\sum_{j = \lfloor \kappa n \rfloor}^{\lfloor \iota n \rfloor-1} e^{n V^{(n)}(j)}}
		{\sum_{j = \lfloor \kappa n \rfloor}^{i-1} e^{n V^{(n)}(j)}}
		& \text{if $\lfloor \iota n \rfloor < i$,}
	\end{cases}
\]
whereas
\[ 
	\mathbb{P}_{i}\left\{T^{(n)}_{\lfloor \kappa n \rfloor} < T^{(n)}_{i+}\right\}
	= \frac{\mu_{i}}{\lambda_{i}+\mu_{i}}
		\mathbb{P}_{i-1}\left\{T^{(n)}_{\lfloor \kappa n \rfloor} < T^{(n)}_{i}\right\}
\]
and
\[
	\mathbb{P}_{i-1}\left\{T^{(n)}_{\lfloor \kappa n \rfloor} < T^{(n)}_{i}\right\}
	= \frac{e^{n V^{(n)}(i-1)}}{\sum_{j = \lfloor \kappa n \rfloor}^{i-1} e^{n V^{(n)}(j)}}.
\]
Thus, for $i > \lfloor \iota n \rfloor$,
\[
	\frac{1}{n (\lambda_{i}+\mu_{i})}
	\frac{\mathbb{P}_{\lfloor \iota n \rfloor} \left\{T^{(n)}_{i} < T^{(n)}_{\lfloor \kappa n \rfloor}\right\}}
	{\mathbb{P}_{i}\left\{T^{(n)}_{\lfloor \kappa n \rfloor} < T^{(n)}_{i+}\right\}}
	= \frac{\sum_{j = \lfloor \kappa n \rfloor}^{\lfloor \iota n \rfloor-1} e^{n V^{(n)}(j)}}
		{n \mu_{i} e^{n V^{(n)}(i-1)}}
	\sim \frac{\sqrt{\frac{\mu(\iota)\lambda\left(\frac{i-1}{n}\right)}
		{\lambda(\iota)\mu\left(\frac{i-1}{n}\right)}}e^{n \left(V(\iota)-V\left(\frac{i-1}{n}\right)\right)}}
	{\mu\left(\frac{i-1}{n}\right) i \left(1-\frac{\lambda(\iota)}{\mu(\iota)}\right)}
\]
since $V$ is minimized at $\kappa$, and, since $\mu(x)$ and $\lambda(x)$ are, respectively, increasing and decreasing, the latter is bounded above by 
\[
	\frac{e^{n \left(V(\iota)-V\left(\frac{i-1}{n}\right)\right)}}{ (\mu(\iota)-\lambda(\iota))\iota}.
\]
Further, 
\[
	V(\iota)-V\left(\frac{i-1}{n}\right) = - V'(z)\left(\frac{i-1}{n}-\iota\right) 
		< -V'(\iota)\left(\frac{i-1}{n}-\iota\right)
\]
for some $z \in \left[\iota,\frac{i-1}{n}\right]$; the inequality follows since $V''(x) > 0$ for all $x$.  Thus,
\[
	\sum_{i = \lfloor \iota n \rfloor}^{\infty}Êe^{n \left(V(\iota)-V\left(\frac{i-1}{n}\right)\right)}
	<  e^{V'(\iota)} \sum_{i = 0}^{\infty}  e^{-V'(\iota)i} = \frac{e^{V'(\iota)}}{1-e^{-V'(\iota)}},
\]
since $\iota > \kappa$ and $V'(\kappa) = 0$. In particular, we see that the sum 
\[
	\sum_{i = \lfloor \iota n \rfloor + 1}^{\infty} \frac{1}{n (\lambda_{i}+\mu_{i})}
	\frac{\mathbb{P}_{\lfloor \iota n \rfloor} \left\{T^{(n)}_{i} < T^{(n)}_{\lfloor \kappa n \rfloor}\right\}}
	{\mathbb{P}_{i}\left\{T^{(n)}_{\lfloor \kappa n \rfloor} < T^{(n)}_{i+}\right\}}
\]
is bounded above.  

Finally, consider
\[
	  \sum_{i = \lfloor \kappa n \rfloor + 1}^{\lfloor \iota n \rfloor} 
	\frac{1}{n \mu_{i} \mathbb{P}_{i-1}\left\{T^{(n)}_{\lfloor \kappa n \rfloor} < T^{(n)}_{i}\right\}}.
\]
Arguing as above, 
\[
	\mathbb{P}_{i-1}\left\{T^{(n)}_{\lfloor \kappa n \rfloor} < T^{(n)}_{i}\right\}
	\sim \frac{1}{1 - \frac{\lambda\left(\frac{i-1}{n}\right)}{\mu\left(\frac{i-1}{n}\right)}}
	\sim \frac{1}{1 - \frac{\lambda\left(\frac{i}{n}\right)}{\mu\left(\frac{i}{n}\right)}},
\]
and, proceeding as in Proposition \ref{CAPACITYTIME}, one can show that
\[
	 \sum_{i = \lfloor \kappa n \rfloor + 1}^{\lfloor \iota n \rfloor} 
	 \frac{1}{\left(\lambda\left(\frac{i}{n}\right)-\mu\left(\frac{i}{n}\right)\right) i}
	 \sim - \frac{1}{(\lambda'(\kappa)-\mu'(\kappa))\kappa}\ln{n}.
\]
\end{proof}

Having established in Proposition \ref{RETURNS} that the process will nonetheless reattain $\iota n$, it is natural to ask how long this will take.  This leads us to:

\begin{prop}
 Let $\kappa < \iota < \eta$. Then,
 \[
	\mathbb{E}_{\lfloor \iota n \rfloor} \left[T^{(n)}_{\lfloor \iota n \rfloor+} 
		\middle\vert T^{(n)}_{\lfloor \iota n \rfloor+} < T^{(n)}_{0}\right]
		\sim \frac{\mu(\iota)}{\mu(\iota) + \lambda(\iota)}\sqrt{\frac{2\pi}
		{n\left(\frac{\mu'(\kappa)}{\mu(\kappa)}-\frac{\lambda'(\kappa)}{\lambda(\kappa)}\right)}
		\frac{\mu(\iota)}{\lambda(\iota)}}
		\frac{e^{n(V(\iota)-V(\kappa))}}{\lambda(\kappa) \kappa}
\]
\end{prop}

\begin{proof}
To begin, we decompose the expectation according to whether, starting from $\lfloor \iota n \rfloor$, the next event is a birth or a death:
\begin{multline*}
	\mathbb{E}_{\lfloor \iota n \rfloor} \left[T^{(n)}_{\lfloor \iota n \rfloor+} 
		\middle\vert T^{(n)}_{\lfloor \iota n \rfloor+} < T^{(n)}_{0}\right]
		= \frac{\lambda_{\lfloor \iota n \rfloor}}
			{\lambda_{\lfloor \iota n \rfloor} + \mu_{\lfloor \iota n \rfloor}}
			\mathbb{E}_{\lfloor \iota n \rfloor + 1} \left[T^{(n)}_{\lfloor \iota n \rfloor} 
				\middle\vert T^{(n)}_{\lfloor \iota n \rfloor} < T^{(n)}_{0}\right]\\
		+ \frac{\mu_{\lfloor \iota n \rfloor}}
			{\lambda_{\lfloor \iota n \rfloor} + \mu_{\lfloor \iota n \rfloor}} 
			\mathbb{E}_{\lfloor \iota n \rfloor - 1} \left[T^{(n)}_{\lfloor \iota n \rfloor} 
				\middle\vert T^{(n)}_{\lfloor \iota n \rfloor} < T^{(n)}_{0}\right]\\
		= \frac{\lambda_{\lfloor \iota n \rfloor}}
			{\lambda_{\lfloor \iota n \rfloor} + \mu_{\lfloor \iota n \rfloor}} 
			\sum_{i = \lfloor \iota n \rfloor + 1}^{\infty}
			\frac{\mathbb{P}_{\lfloor \iota n \rfloor + 1} 
				\left\{T^{(n)}_{i} < T^{(n)}_{\lfloor \iota n \rfloor} 
				\middle\vert T^{(n)}_{\lfloor \iota n \rfloor} < T^{(n)}_{0}\right\}}
				{n \mu_{i} \mathbb{P}_{i - 1} 
				\left\{T^{(n)}_{\lfloor \iota n \rfloor}  < T^{(n)}_{i} 
				\middle\vert T^{(n)}_{\lfloor \iota n \rfloor} < T^{(n)}_{0}\right\}}\\
		+ \frac{\mu_{\lfloor \iota n \rfloor}}
			{\lambda_{\lfloor \iota n \rfloor} + \mu_{\lfloor \iota n \rfloor}} 
			\sum_{i=1}^{\lfloor \iota n \rfloor - 1}
			\frac{\mathbb{P}_{\lfloor \iota n \rfloor - 1} 
				\left\{T^{(n)}_{i} < T^{(n)}_{\lfloor \iota n \rfloor} 
				\middle\vert T^{(n)}_{\lfloor \iota n \rfloor} < T^{(n)}_{0}\right\}}
				{n \lambda_{i} \mathbb{P}_{i + 1} 
				\left\{T^{(n)}_{\lfloor \iota n \rfloor}  < T^{(n)}_{i} 
				\middle\vert T^{(n)}_{\lfloor \iota n \rfloor} < T^{(n)}_{0}\right\}}.
\end{multline*}

For the first sum, we observe that for any $i \geq \lfloor \iota n \rfloor$, 
\[
	\mathbb{P}_{i}\left\{T^{(n)}_{\lfloor \iota n \rfloor} < T^{(n)}_{0}\right\} = 1,
\]
and we may thus replace the conditional probabilities with the unconditional ones. Then, using \eqref{henV},
\[
	\frac
	{\mathbb{P}_{\lfloor \iota n \rfloor + 1}\left\{T^{(n)}_{i} < T^{(n)}_{\lfloor \iota n \rfloor}\right\}}
	{\mathbb{P}_{i - 1}\left\{T^{(n)}_{\lfloor \iota n \rfloor}  < T^{(n)}_{i} \right\}}
	= e^{n \left(V^{(n)}(\lfloor \iota n \rfloor) - V^{(n)}(i)\right)},
\]
so that, using Lemma \ref{LAPLACE}, the first sum is asymptotic to 
\[
	\frac{\mu(\iota)}{\mu(\iota) + \lambda(\iota)} \frac{1}{(\mu(\iota) - \lambda(\iota))\iota}.
\]

For the second sum, we observe that 
\[
	\left\{T^{(n)}_{\lfloor \iota n \rfloor} < T^{(n)}_{i}\right\} \cap 
		\left\{ T^{(n)}_{\lfloor \iota n \rfloor} < T^{(n)}_{0}\right\}
		= \left\{T^{(n)}_{\lfloor \iota n \rfloor} < T^{(n)}_{i}\right\},
\]
whereas 
\[
	\mathbb{P}_{\lfloor \iota n \rfloor - 1} 
		\left\{T^{(n)}_{i} < T^{(n)}_{\lfloor \iota n \rfloor}, 
			T^{(n)}_{\lfloor \iota n \rfloor} < T^{(n)}_{0}\right\}
	= \mathbb{P}_{\lfloor \iota n \rfloor - 1}\left\{T^{(n)}_{i} < T^{(n)}_{\lfloor \iota n \rfloor}\right\}
		\mathbb{P}_{i}\left\{T^{(n)}_{\lfloor \iota n \rfloor} < T^{(n)}_{0}\right\},
\]
so that, applying Bayes' theorem,
\[
	\frac{\mathbb{P}_{\lfloor \iota n \rfloor - 1} \left\{T^{(n)}_{i} < T^{(n)}_{\lfloor \iota n \rfloor} 
		\middle\vert T^{(n)}_{\lfloor \iota n \rfloor} < T^{(n)}_{0}\right\}}
		{\mathbb{P}_{i + 1} \left\{T^{(n)}_{\lfloor \iota n \rfloor}  < T^{(n)}_{i} 
				\middle\vert T^{(n)}_{\lfloor \iota n \rfloor} < T^{(n)}_{0}\right\}}
	= \frac{\mathbb{P}_{\lfloor \iota n \rfloor - 1}
		\left\{T^{(n)}_{i} < T^{(n)}_{\lfloor \iota n \rfloor}\right\}
		\mathbb{P}_{i}\left\{T^{(n)}_{\lfloor \iota n \rfloor} < T^{(n)}_{0}\right\}}
		{\mathbb{P}_{i+1}\left\{T^{(n)}_{\lfloor \iota n \rfloor} < T^{(n)}_{i}\right\}}.
\]
Again, from \eqref{henV}, we see that 
\[
	\frac{\mathbb{P}_{\lfloor \iota n \rfloor - 1}
		\left\{T^{(n)}_{i} < T^{(n)}_{\lfloor \iota n \rfloor}\right\}}
		{\mathbb{P}_{i+1}\left\{T^{(n)}_{\lfloor \iota n \rfloor} < T^{(n)}_{i}\right\}}
		= e^{n V^{(n)}(\lfloor \iota n \rfloor-1) - V^{(n)}(i))},
\]
so this sum reduces to 
\[ 
	\sum_{i=1}^{\lfloor \iota n \rfloor - 1} 
		\frac{\mathbb{P}_{i}\left\{T^{(n)}_{\lfloor \iota n \rfloor} < T^{(n)}_{0}\right\}}{n \lambda_{i}}
			e^{n V^{(n)}(\lfloor \iota n \rfloor-1) - V^{(n)}(i))}.
\]

To evaluate the sum, it is useful to consider it in two pieces.  To do so, we first re-introduce $\iota' < \kappa$ such that $V(\iota') = V(\iota)$, and then consider
\begin{multline*}
	\sum_{i=1}^{\lfloor \iota' n \rfloor - 1} 
		\frac{\mathbb{P}_{i}\left\{T^{(n)}_{\lfloor \iota n \rfloor} < T^{(n)}_{0}\right\}}{n \lambda_{i}}
			e^{n V^{(n)}(\lfloor \iota n \rfloor-1) - V^{(n)}(i))}\\
	+ \sum_{i=\lfloor \iota' n \rfloor}^{\lfloor \iota n \rfloor - 1}
		\frac{\mathbb{P}_{i}\left\{T^{(n)}_{\lfloor \iota n \rfloor} < T^{(n)}_{0}\right\}}{n \lambda_{i}}
			e^{n V^{(n)}(\lfloor \iota n \rfloor-1) - V^{(n)}(i))}.
\end{multline*}
For the former, $V^{(n)}(\lfloor \iota n \rfloor-1) - V^{(n)}(i))$ is maximized at $i = \lfloor \iota' n \rfloor-1$, whereas $\mathbb{P}_{i}\left\{T^{(n)}_{\lfloor \iota n \rfloor} < T^{(n)}_{0}\right\}$ us bounded above by 1.   Using Lemma \ref{LAPLACE}, the first sum is asymptotically bounded above by
\[
	 \frac{1}{\left(\lambda\left(\frac{\lfloor \iota' n \rfloor-1}{n}\right)
	 - \mu\left(\frac{\lfloor \iota' n \rfloor-1}{n}\right)\right)(\lfloor \iota' n \rfloor-1)}
	 	\sqrt{\frac{\mu\left(\frac{\lfloor \iota n \rfloor-1}{n}\right)
			\lambda\left(\frac{\lfloor \iota' n \rfloor-1}{n}\right)}
			{\lambda\left(\frac{\lfloor \iota n \rfloor-1}{n}\right)
			\mu\left(\frac{\lfloor \iota' n \rfloor-1}{n}\right)}}.
\]
For the second piece, we note that for $\lfloor \iota' n \rfloor \leq i < \lfloor \iota' n \rfloor$, 
$\mathbb{P}_{i}\left\{T^{(n)}_{\lfloor \iota n \rfloor} < T^{(n)}_{0}\right\} \sim 1$, whereas 
$V^{(n)}(\lfloor \iota n \rfloor-1) - V^{(n)}(i))$ is maximized at $\lfloor \kappa n \rfloor$, so appealing to Lemma \ref{LAPLACE}, it is asymptotically equivalent to 
\[
	\sqrt{\frac{2\pi}
		{n\left(\frac{\mu'(\kappa)}{\mu(\kappa)}-\frac{\lambda'(\kappa)}{\lambda(\kappa)}\right)}
		\frac{\mu(\iota)}{\lambda(\iota)}}
		\frac{e^{n(V(\iota)-V(\kappa))}}{\lambda(\kappa) \kappa}
\]
The result follows.
\end{proof}

Finally, for the sake of completeness, we consider the extinction time conditioned on never reaching carrying capacity; moreover, for the sake of variety, we take a different approach to the proof and consider the process conditioned on never hitting $\lfloor \kappa n \rfloor$:

\begin{lem}
The logistic process conditioned on the event $T^{(n)}_{0} < T^{(n)}_{M}$ is a Markov birth and death process with transition rates
\[
	\tilde{\lambda}^{(n)}_{i} = \lambda^{(n)}_{i}\frac{h^{(n)}_{0,M}(i+1)}{h^{(n)}_{0,M}(i)}
	\quad \text{and} \quad 
	\tilde{\mu}^{(n)}_{i} = \mu^{(n)}_{i}\frac{h^{(n)}_{0,M}(i-1)}{h^{(n)}_{0,M}(i)},
\]
In particular, taking $M = \lfloor \iota n \rfloor$ for $\kappa < \iota < \eta$, we have that 
\[
	\lim_{n \to \infty} \tilde{\lambda}^{(n)}_{i} = \mu(0) i
	\quad \text{and} \quad 
	\lim_{n \to \infty} \tilde{\mu}^{(n)}_{i} = \lambda(0) i
\]
\end{lem}

\begin{proof}
The first statement is an easy consequence of Bayes' theorem; by definition, the conditioned process $\tilde{X}^{(n)}$ has transition rates
\begin{align*}
	q^{(n)}_{ij} &= \lim_{h \downarrow 0} \frac{1}{h}
	\mathbb{P}\left\{ \tilde{X}^{(n)}(t+h) = j \middle\vert \tilde{X}^{(n)}(t+h) = i\right\}\\
	&= \lim_{h \downarrow 0} \frac{1}{h} 
	\mathbb{P}_{i}\left\{X^{(n)}(t+h) = j \middle\vert T^{(n)}_{0} < T^{(n)}_{M}\right\}\\
	&= \lim_{h \downarrow 0} \frac{1}{h} 
	\frac{\mathbb{P}_{i}\left\{X^{(n)}(t+h) = j , T^{(n)}_{0} < T^{(n)}_{M}\right\}}
		{\mathbb{P}_{i}\left\{T^{(n)}_{0} < T^{(n)}_{M}\right\}}\\
\intertext{which, by the Markov property, is}
	&= \lim_{h \downarrow 0} \frac{1}{h} 
	\frac{\mathbb{P}_{i}\left\{X^{(n)}(t+h) = j\right\}
		\mathbb{P}_{j}\left\{T^{(n)}_{0} < T^{(n)}_{M}\right\}}
		{\mathbb{P}_{i}\left\{T^{(n)}_{0} < T^{(n)}_{M}\right\}}\\
	&= \begin{cases}
		\lambda^{(n)}_{i}\frac{h^{(n)}_{0,M}(i+1)}{h^{(n)}_{0,M}(i)} & \text{if $j = i+1$,}\\
		\mu^{(n)}_{i}\frac{h^{(n)}_{0,M}(i-1)}{h^{(n)}_{0,M}(i)} & \text{if $j = i-1$, and}\\
		0 & \text{otherwise.}
	\end{cases}
\end{align*}
The second statement follows immediately from Proposition \ref{EXTINCTION}.
\end{proof}

\begin{rem} This is simply a special case of Doob's $h$-transform \cite{Doob1957}, but we include it  in keeping with our aim of a largely self-contained, elementary treatment.
\end{rem}

\begin{lem}
Let 
\[
	\tau^{(n)}_{M}(m) = \mathbb{E}_{m}\left[T^{(n)}_{0} \middle\vert T^{(n)}_{0} < T^{(n)}_{M}\right].
\]
Then,
\[
	\tau^{(n)}_{M}(m) = \sum_{i=1}^{m} \sum_{j=i}^{M-1} \frac{1}{\tilde{\lambda}^{(n)}_{j}} 
	\prod_{k=i}^{j} \frac{\tilde{\lambda}^{(n)}_{k} }{\tilde{\mu}^{(n)}_{k}}
\]
\end{lem}

\begin{proof}
For $m < M$ the function $\tau^{(n)}_{M}$ satisfies the recurrence relation
\[
	\tau^{(n)}_{M}(m) = \frac{1}{\tilde{\lambda}^{(n)}_{m} + \tilde{\mu}^{(n)}_{m}} 
	+ \frac{\tilde{\lambda}^{(n)}_{m}}{\tilde{\lambda}^{(n)}_{m} 
	+ \tilde{\mu}^{(n)}_{m}} \tau^{(n)}_{M}(m+1) 
	+ \frac{\tilde{\mu}^{(n)}_{m}}{\tilde{\lambda}^{(n)}_{m} 
	+ \tilde{\mu}^{(n)}_{m}} \tau^{(n)}_{M}(m-1),
\]
with boundary $\tau^{(n)}_{M}(0) = 0$, whilst
\[
	\tau^{(n)}_{M}(M-1) = \frac{1}{\tilde{\mu}^{(n)}_{M-1}} + \tau^{(n)}_{M}(M-2).
\]
Solving the recurrence equation gives the result. As previously, we refer to \cite{Karlin1975} for a detailed treatment.   
\end{proof}

\begin{prop}
Let $\lambda(x)$ and $\mu(x)$ be continuous and fix $m < \lfloor \kappa n \rfloor$.  Then,
\[
	\lim_{n \to \infty} \mathbb{E}_{m}\left[T^{(n)}_{0} \middle\vert 
		T^{(n)}_{0} < T^{(n)}_{\lfloor \kappa n \rfloor}\right] = 
 		-\frac{1}{\mu(0)} \ln\left(1-\frac{\mu(0)}{\lambda(0)}\right).
\]
\end{prop}

\begin{proof}
Proceeding as in Propositon \ref{EXTINCTION}, we can show that we can interchange sums and limits to obtain 
\begin{multline*}
	\lim_{n \to \infty} \tau^{(n)}_{\lfloor \kappa n\rfloor}(m) 
	= \sum_{j=1}^{\infty} \frac{1}{\mu(0) j} \left(\frac{\mu(0)}{\lambda(0)}\right)^{j}
	= \frac{1}{\mu(0)} \sum_{j=1}^{\infty} \int_{0}^{\frac{\mu(0)}{\lambda(0)}} x^{j-1}\, dx\\
	 = \frac{1}{\mu(0)} \int_{0}^{\frac{\mu(0)}{\lambda(0)}} \frac{dx}{1-x} 
	 = -\frac{1}{\mu(0)} \ln\left(1-\frac{\mu(0)}{\lambda(0)}\right).
\end{multline*}

\end{proof}

\section{Discussion}

Propositions \ref{EXTINCTION} and \ref{REACH} provide an elementary proof for the observation that the probability that a density-limited population successfully invades an unoccupied territory, in the limit as carrying capacity tends to infinity, is essentially the survival probability of a suitably chosen branching process, and that, assuming the process reaches $m_{n}$ for \textit{any} sequence $\{m_{n}\}$ such that  $m_{n} \to \infty$, its attainment of much higher levels is assured.   This has previously been asserted without proof and defended heuristically, \eg \cite{Iwasa2004,Desai2007,Weissman2009} or rigorously proven rigorously via coupling arguments, \eg \cite{Andersson+Djehiche98,Champagnat2006b,Parsons2012}.  A similar approach to ours has recently appeared in \cite{Chalub2015}, where asymptotic expansions of Laplace integrals \cite{Bender+Orszag78} were applied to analyze \eqref{H} via the asymptotic expression \label{LAPLACE} in the case when $\omega \leq 1$; this leads to a slightly different expression involving exponentials for the fixation probability, as this approach implicitly passes from the discrete branching process to Feller's diffusion approximation \cite{Feller1951}. 

While our approach lacks the intuition for the pathwise behaviour obtained via coupling with a branching process, in addition to its simplicity, it has the advantage of allowing us to make assertions about the behaviour of the logistic process for population sizes considerably larger than those for which the coupling remains exact (with high probability, the stochastic logistic process and the  branching process will diverge once the population size has exceeded $\BigO{\sqrt{n}}$ individuals, see \cite{Andersson+Djehiche98}).

As such, we can make observations about the maximum of the logistic process without having to resort to large deviations arguments, and over the lifetime of the population, rather than being limited to a compact time interval.   In particular, we see the initially surprising fact that a population, once it has invaded (for present purposes, this means to attain a population size $m_{n}$ such that $m_{n} \to \infty$, \eg a positive fraction of carrying capacity), it will, with high probability, greatly exceed carrying capacity -- potentially reaching twice carrying capacity -- prior to extinction.  To obtain some intuition for this fact, consider briefly the large numbers approximation, \eqref{LLN}: if we linearize about the carrying capacity, $\kappa$, the deterministic dynamics are locally symmetric, with a restoring force pushing back towards carrying capacity proportional to the size of the fluctuation and the dominant eigenvalue at $x = \kappa$, $\lambda'(\kappa) - \mu'(\kappa)$ (hence the symmetric Gaussian fluctuations in \eqref{CLT}).  Of necessity, the process will eventually make a large fluctuation down to zero, and numerous ``failed attempts'' en route.  The symmetry near carrying capacity essentially ensures that it will equally likely to make large upward fluctuations as well.

To conclude, we note that in addition their mathematical and theoretical interest our results have  potentially practically important biological implications: depending on when we observe a population, it may in fact be far in excess of equilibrium, so that our estimates of the population viability or Mathusian growth rate may be grossly inaccurate.  Moreover, large increases or sudden decreases in population size, that might be interpreted as indicators of recovery or collapse, may be little more than the  effects of demographic stochasticity.  Some caveats are, of course in order.  First of all, as we observe in Remark \ref{TWICE}, whilst significant fluctuations above carrying capacity (\ie of magnitude proportional to the population size) are highly likely to occur, the size of those fluctuations is nonetheless model-dependent, and in a more realistic model, we expect many factors will limit population growth at high frequencies (\eg resource limitation, mate competition, disease, and predation to name a few) so that the symmetry discussed in the previous paragraph is truly a near carrying capacity phenomenon, and the potential barrier close to carrying capacity, so that fluctuations to twice carrying capacity might be highly improbable outside of toy mathematical models.  Similarly, including Allee effects at small numbers would qualitatively change the shape of the potential, adding a second well at small numbers and a potential barrier to traverse to reach the carrying capacity, which would demand a more detailed analysis than is presented here, although one likely amenable an analysis similar to that for diffusions with metastable states (\eg\cite{Bovier2004}).  

More importantly, as we see from Propositions \ref{EXTINCTIONTIME} and \ref{EXCURSIONTIME}, the expected waiting time for an excursion far above carrying capacity, despite being asymptotically smaller than the extinction time, is still exponentially large in the population size $n$, so that one may have to potentially wait an extremely long time to see such an excursion, even if arbitrarily many such excursions may occur prior to extinction (Proposition \ref{RETURNS}).  Moreover, as we see from Propositions \ref{CAPACITYTIME} and \ref{RETURNTIME}, the time to return time to equilibrium is considerably shorter than the time between excursion -- so for large populations, we are extremely unlikely to observe the population far from its equilibrium size.  Indeed, from Proposition \ref{EXTINCTIONTIME} and Corollary \ref{EXECESSTIME}, the fraction of time spent above $\lfloor \iota n \rfloor$, for any $\iota > \kappa$, is of order $e^{n \left(V(\kappa) - V(\iota)\right)}$, which for large populations is vanishingly small.

Nonetheless, when modelling populations \eg in performing population viability analysis, we are often most interested in the smaller populations of uncommon species which are at risk of extinction, in which case these large fluctuations may well occur on the timescale at which we observe the population and confound our efforts to estimate extinction risk.  In such cases, the existence of stochastic fluctuations of large magnitude that we have demonstrated should serve as an important caution to the use of deterministic modelling in conservation biology.  
 
 \section*{Acknowledgements}
I thank Fran\c{c}ois Bienvenu, Amaury Lambert, Peter Ralph and Tim Rogers for commenting on a draft of this manuscript and for giving some intuition for the results therein.

\appendix
    
\section{Some Useful Lemmas}  

For the sake of completeness, we include some simple and well-known lemmas that we have used in the main text.

\begin{lem}\label{INEQ}
\begin{enumerate}[(i)]
\item Let $f$ be an non-decreasing function.  Then,
\[
	0 \leq \frac{1}{n} \sum_{j=a+1}^{b} f\left(\frac{j}{n}\right) - \int_{\frac{a}{n}}^{\frac{b}{n}} f(x)\, dx
	\leq \frac{1}{n} \left(f\left(\frac{b}{n}\right) - f\left(\frac{a}{n}\right)\right)
\]
and
\[
	0 \leq \int_{\frac{a}{n}}^{\frac{b}{n}} f(x)\, dx - \frac{1}{n} \sum_{j=a}^{b-1} f\left(\frac{j}{n}\right) 
	\leq \frac{1}{n} \left(f\left(\frac{b}{n}\right) - f\left(\frac{a}{n}\right)\right).
\]
\item If $f$ is differentiable and $|f'(x)|$ is bounded by $M$ on $\left[\frac{a}{n},\frac{b}{n}\right]$, then 
\[
	\left|\int_{\frac{a}{n}}^{\frac{b}{n}} f(x)\, dx - \frac{1}{n} \sum_{j=a}^{b-1} f\left(\frac{j}{n}\right)\right|
	< \frac{M(b-a)}{2 n^{2}}.
\]
\item Further, if $f$ is twice differentiable and $|f''(x)|$ is bounded by $M$ on $\left[\frac{a}{n},\frac{b}{n}\right]$, then 
\[
	\left|\int_{\frac{a}{n}}^{\frac{b}{n}} f(x)\, dx - \frac{1}{n} \sum_{j=a}^{b-1} f\left(\frac{j}{n}\right)
	-\frac{1}{2n}\left(f\left(\frac{b}{n}\right) - f\left(\frac{a}{n}\right)\right)\right|
	< \frac{M(b-a)}{4 n^{3}}.
\]
\end{enumerate}
\end{lem}



\begin{proof}
We will prove the first, third and fourth inequalities; the proof of the second is identical to the first.  First,
\[
	\int_{\frac{a}{n}}^{\frac{b}{n}} f(x)\, dx = \sum_{j=a+1}^{b} \int_{\frac{j-1}{n}}^{\frac{j}{n}} f(x)\, dx.
\]
Thus, if $f$ is increasing,
\[
	 \frac{1}{n} f\left(\frac{j-1}{n}\right) \leq \int_{\frac{j-1}{n}}^{\frac{j}{n}} f(x)\, dx
	 \leq \frac{1}{n} f\left(\frac{j}{n}\right),
\]
so
\[
	\frac{1}{n} \sum_{j=1}^{i} f\left(\frac{j-1}{n}\right) \leq \int_{0}^{\frac{i}{n}} f(x)\, dx 
	\leq \frac{1}{n} \sum_{j=1}^{i} f\left(\frac{j}{n}\right).
\]	
and
\begin{multline*}
	0 \leq \frac{1}{n} \sum_{j=a+1}^{b} f\left(\frac{j}{n}\right) - \int_{\frac{a}{n}}^{\frac{b}{n}} f(x)\, dx\\
	\leq \frac{1}{n} \sum_{j=a+1}^{b} f\left(\frac{j}{n}\right) 
		- \sum_{j=a+1}^{b} f\left(\frac{j-1}{n}\right)
	=  \frac{1}{n} \left(f\left(\frac{b}{n}\right) - f\left(\frac{a}{n}\right)\right).
\end{multline*}

If on the other hand, $f$ is differentiable and $|f'(x)| < M$,
\[
	 \int_{\frac{j-1}{n}}^{\frac{j}{n}} f(x)\, dx - \frac{1}{n} f\left(\frac{j}{n}\right) 
	  =  \int_{\frac{j-1}{n}}^{\frac{j}{n}} f(x) - f\left(\frac{j}{n}\right)\, dx,
\]
and, by the mean value theorem, there exists $\xi_{j} \in \left[\frac{j-1}{n},\frac{j}{n}\right]$ such that 
\[
	f\left(\frac{j}{n}\right)- f(x) = f'(\xi_{j})  {\textstyle \left(x-\frac{j}{n}\right)},
\]
and thus 
\[
	\int_{\frac{j-1}{n}}^{\frac{j}{n}} f(x)\, dx - \frac{1}{n} f\left(\frac{j}{n}\right)
	= f'(\xi_{j}) \int_{\frac{j-1}{n}}^{\frac{j}{n}} {\textstyle \left(x-\frac{j}{n}\right)}\, dx 
	= \frac{f'(\xi_{j})}{2 n^{2}} 
\] 
Thus, the absolute value is bounded by $\leq \frac{M}{2 n^{2}}$. Summing over $j=a,\ldots,b-1$ gives the result.

For the final statement, we take $\xi_{j}$ as previously, and observe that proceeding as above, one obtains that 
\[
	\left|\frac{1}{n} \sum_{j=a+1}^{b} f'(\xi_{j}) - \int_{\frac{a}{n}}^{\frac{b}{n}} f'(x)\, dx\right|
	\leq \frac{M}{2 n^{2}}\left(\frac{b}{n}-\frac{a}{n}\right)
\] 

Since 
\[
	\int_{\frac{a}{n}}^{\frac{b}{n}} f(x)\, dx - \frac{1}{n} \sum_{j=a}^{b-1} f\left(\frac{j}{n}\right)
	= \frac{1}{2n^2} \sum_{j=a+1}^{b} f'(\xi_{j})
\]
and 
\[
	\int_{\frac{a}{n}}^{\frac{b}{n}} f'(x)\, dx = f\left(\frac{b}{n}\right) - f\left(\frac{a}{n}\right),
\]
the result follows.
\end{proof}

\begin{lem}[Dominated convergence theorem for series] \label{DCTS}
Suppose that $a_{m,n}$ and $b_{m}$ are sequences such that 
\begin{enumerate}[(i)]
\item $a_{m,n} \to a_{m}$ as $n \to \infty$,
\item $|a_{m,n}| \leq b_{m}$ for all $n$, and
\item $\sum_{m=0}^{\infty} b_{m} < \infty$.
\end{enumerate}
Then,
\[
	\lim_{n \to \infty} \sum_{m = 0}^{\infty} a_{m,n} = \sum_{m = 0}^{\infty} a_{m},
\]
\ie the sums on the left and right are convergent and one can interchange sum and limit.
\end{lem}

\begin{proof}
The existence and convergence of the $b_{m}$ ensure absolute convergence of the series
\[
	\sum_{m = 0}^{\infty} a_{m,n}
	\quad \text{and} \quad
	\sum_{m = 0}^{\infty} a_{m}.
\]
To see that one may interchange the sum and limit, fix $\varepsilon > 0$, and choose $M$ so that 
\[
	\sum_{m=M+1}^{\infty} b_{m} < \frac{\varepsilon}{3}.
\]
For each $m=0,\ldots,M$, choose $N_{m}$ such that $|a_{m,n}-a_{m}| < \frac{\varepsilon}{3M}$ for all $n \geq N_{m}$ and let $N = \max_{1 \leq m \leq M} N_{m}$.  Then, for $n \geq N$,
\[Ê	
	\abs{\sum_{m = 0}^{\infty} a_{m,n} - \sum_{m = 0}^{\infty} a_{m}}
	\leq \sum_{m = 0}^{M} |a_{m,n} - a_{m}|
		+ \sum_{m = M+1}^{\infty} |a_{m,n}| + \sum_{m = M+1}^{\infty} |a_{m}| < \varepsilon.
\]
\end{proof}

\begin{lem}\label{SUBSEQ}
Suppose that $a_{m,n}$ and $a_{m}$ are sequences such that 
\begin{enumerate}[(i)]
\item for each $m$, $a_{m,n} \to a_{m}$ as $n \to \infty$, and
\item $a_{m} \to a$ as $m \to \infty$.
\end{enumerate}
Then, there exists a sequence $m_{n}$ such that $m_{n} \to \infty$  and $a_{m_{n}} \to a$ as  $n \to \infty$.
\end{lem}

\begin{proof}
Fix $\varepsilon > 0$.  Since  $a_{m} \to a$ as $m \to \infty$, there exists $M > 0$ such that 
\[
	|a_{m} - a| < \frac{\varepsilon}{2}
\]
for all $m > M$.  Moreover, since $a_{m,n} \to a_{m}$ as $n \to \infty$, for each $m$ there exists an $n_{m}$ such that 
\[
	|a_{m,n} - a_{m}| < \frac{\varepsilon}{2}
\]
for all $n > n_{m}$.  Without loss of generality, we may choose $n_{m} > n_{m-1}$, so that $n_{m} \to \infty$ as $m \to \infty$.  

Now, for each $n$, we define
\[
	m_{n} = \sup_{m} \{n_{m} < n\}.
\]
Since $n_{m} \to \infty$ as $m \to \infty$, similarly, $m_{n} \to \infty$ as $n \to \infty$, so for $n$ sufficiently large, we have $m_{n} > M$.  Moreover, we have $n_{m_{n}} < n$ so that 
\[
	|a_{m_{n},m} - a| \leq |a_{m_{n},n} - a_{m_{n},n}| + |a_{m_{n}} - a| < \varepsilon.
\]
\end{proof}

\begin{lem}\label{DROP}
Suppose that $a_{i}, b_{i} \geq 0$ for all $i$, that
\[
	\lim_{i \to \infty} \frac{a_{i}}{b_{i}} = 1,
\]
and that 
\[
	\lim_{n \to \infty} \sum_{i = 1}^{n} b_{i} = \infty.
\]
Then,
\[
	\lim_{n \to \infty} \frac{\sum_{i = 1}^{n} a_{i}}{\sum_{i = 1}^{n} b_{i}} = 1.
\]
\end{lem}

\begin{proof}
Fix $\varepsilon > 0$, and choose $M > 0$ such that 
\[
	 \left|\frac{a_{i}}{b_{i}} - 1 \right| < \frac{\varepsilon}{2},
\]
for all $m > M$.  Next, since $\sum_{i = 1}^{n} b_{i}$ diverges, we may choose $N > 0$ such that 	
\[
	\left|\sum_{i = 1}^{M} a_{i}\right| < \frac{\varepsilon}{2} \sum_{i = M}^{n} b_{i}
	\quad \text{and} \quad
	\left|\sum_{i = 1}^{M} b_{i}\right| < \frac{\varepsilon}{2} \sum_{i = M}^{n} b_{i}
\]
for all $n > N$.  Then,
\begin{multline*} 
	(1 - \varepsilon) \sum_{i = 1}^{n} b_{i}
	\leq \frac{\left(1 - \frac{\varepsilon}{2}\right)}
		{\left(1 + \frac{\varepsilon}{2}\right)} \sum_{i = 1}^{n} b_{i}
	 \leq \left(1 - \frac{\varepsilon}{2}\right) \sum_{i = M+1}^{n} b_{i}
	 \leq \sum_{i = M+1}^{n} a_{i} \\
	 \leq \sum_{i = 1}^{n} a_{i}
	 \leq \sum_{i = M+1}^{n} a_{i} + \left(1 + \frac{\varepsilon}{2}\right) \sum_{i = M+1}^{n} b_{i}  
	 \leq  (1 + \varepsilon) \sum_{i = M+1}^{n} b_{i} \leq (1 + \varepsilon) \sum_{i = 1}^{n} b_{i}, 
\end{multline*}
and the result follows.
\end{proof}

\begin{rem}
We note from the proof that the result applies equally well with sequences $a_{n,i}$ and $b_{n,i}$ depending on $n$, provided $M$ may be chosen independently of $n$, which is the case when one has
 \[
	\lim_{i \to \infty} \sup_{n} \frac{a_{n,i}}{b_{n,i}} = 1.
\]
\end{rem}

\clearpage

\bibliography{Excursions}
\bibliographystyle{plain}

\end{document}